\documentclass{vldb}
\usepackage{booktabs} 

\usepackage{algorithm}
\usepackage[noend]{algorithmic}
\usepackage{graphicx}
\usepackage{multirow}
\usepackage{amssymb}
\usepackage[flushleft,alwaysadjust]{paralist}
\usepackage{color}
\usepackage{url}
\usepackage[T1,OT1]{fontenc}
\usepackage{relsize}
\usepackage{tikz}
\usepackage{subfig}
\usepackage{afterpage}
\usepackage{colortbl}
\usepackage{balance}
\usepackage{alltt}
\usepackage{listings}
\usepackage{enumitem}

\usetikzlibrary{arrows}
\tikzstyle{block}=[draw opacity=0.7,line width=1.4cm]
\usetikzlibrary{positioning}
\usetikzlibrary{arrows, decorations.pathmorphing}
\tikzset{snake arrow/.style=
{-triangle 45,
line width=1.4pt,
decorate,
decoration={snake,amplitude=1mm,segment length=10mm,post length=2mm}},
}

\DeclareCaptionType{copyrightbox}

\newcommand{\darwin}{\textsc{Darwin}}
\newcommand{\darwinp}[1]{\textsc{Darwin}(#1)}
\newcommand{\babystep}{\texttt{LocalSearch}}
\newcommand{\universal}{\texttt{UniversalSearch}}
\newcommand{\hybrid}{\texttt{HybridSearch}}

\newcommand{\behzad}[1]{{\small\color{blue}[(BG)~#1]}}
\newcommand{\sg}[1]{{\small\color{blue}[(SG)~#1]}}
\newcommand{\wctan}[1]{{\small\color{red}[#1]~(wang-chiew)}}

\newcommand{{\kempegreedy}}{G{\scriptsize{REEDY}}}
\newcommand{{\ourgreedy}}{ScoreG{\scriptsize{REEDY}}}
\newtheorem{problem}{Problem}
\newtheorem{definition}{Definition}
\newtheorem{theorem}{Theorem}
\newtheorem{lemma}{Lemma}

\newtheorem{example}{Example}
\newcommand{\spara}[1]{\smallskip\noindent{\bf{#1}}}
\newcommand{\mpara}[1]{\medskip\noindent{\bf{#1}}}

\newlength{\oldtextfloatsep}
\newlength{\oldfloatsep}

\usepackage{xparse} 
\NewDocumentCommand{\rot}{O{45} O{1em} m}{\makebox[#2][l]{\rotatebox{#1}{#3}}}

\vldbTitle{Adaptive Rule Discovery for Labeling Text Data}
\vldbAuthors{Sainyam Galhotra, Behzad Golshan and Wang-Chiew Tan}
\vldbDOI{https://doi.org/10.14778/xxxxxxx.xxxxxxx}
\vldbVolume{12}
\vldbNumber{xxx}
\vldbYear{2019}

\begin{document}

\title{Adaptive Rule Discovery for Labeling Text Data}

\numberofauthors{3} 
\author{
\alignauthor
Sainyam Galhotra\\
       \affaddr{UMass Amherst}\\
       \email{sainyam@cs.umass.edu}
\alignauthor
Behzad Golshan\\
       \affaddr{Megagon Labs}\\
       \email{behzad@megagon.ai}
\alignauthor 
Wang-Chiew Tan\\
       \affaddr{Megagon Labs}\\
       \email{wangchiew@megagon.ai}
}

\maketitle

\begin{abstract}
Creating and collecting labeled data is one of the major bottlenecks in machine
learning pipelines and the emergence of automated feature generation techniques
such as deep learning, which typically requires a lot of training data, has
further exacerbated the problem. While weak-supervision techniques have
circumvented this bottleneck, existing frameworks either require users to write
a set of diverse, high-quality rules to label data (e.g., Snorkel), or require
a labeled subset of the data to automatically mine rules (e.g., Snuba). The
process of manually writing rules can be tedious and time consuming. At the
same time, creating a labeled subset of the data can be costly and even
infeasible in imbalanced settings. This is due to the fact that a random
sample in imbalanced settings often contains only a few positive instances.

To address these shortcomings, we present \darwin, an interactive system
designed to alleviate the task of writing rules for labeling text data in
weakly-supervised settings. Given an initial labeling rule, \darwin\
automatically generates a set of candidate rules for the labeling task at hand,
and utilizes the annotator's feedback to adapt the candidate rules. We describe
how \darwin\ is scalable and versatile.  It can operate over large text corpora
(i.e., more than 1 million sentences) and supports a wide range of labeling
functions (i.e., any function that can be specified using a context free grammar).  
Finally, we demonstrate with a suite of experiments over five real-world datasets
that \darwin\ enables annotators to generate weakly-supervised labels efficiently
and with a small cost. In fact, our experiments show that rules discovered by
\darwin\ on average identify 40\% more positive instances compared to Snuba even
when it is provided with 1000 labeled instances.
\end{abstract}


\section{Introduction}
Today, many applications are powered by machine learning techniques. The
success of {\em deep learning} methods in domains such as natural language
processing and computer vision is further fuelling this trend. While deep
learning (and machine learning in general) can offer superior performance,
training such systems typically requires a large set of labeled examples, which
is expensive and time-consuming to obtain.

{\em Weak supervision} techniques circumvent the above problem to some extent,
by leveraging heuristic rules that can generate (noisy) labels for a subset of
data\footnote{To be more concise, we refer to heuristic rules simply as
heuristics or rules as well
throughout the paper.}. A large volume of labels can be obtained at a low cost this way, and to
compensate for the noise, noise-aware techniques can be used for further
improving the performance of machine learning models \cite{snorkel, snorkel2}.
However, obtaining high-quality labeling heuristics remains a challenging
problem. A subset of existing frameworks, with {\em Snorkel} \cite{snorkel}
being the most notable example, rely on domain experts to provide a set of
labeling heuristics which can be a tedious and time-consuming task. In
contrast, other frameworks aim to automatically mine useful heuristics using
further provided supervision. For instance, Snuba \cite{Snuba} circumvents
dependence on domain experts by requiring a labeled subset of the data, and
then utilizing it to automatically derive labeling heuristics. {\em Babble
labble} is another example which asks expert to label a few examples and
explain their choice (in natural language). This explanation is used to derive
labeling heuristics. While these approaches have been quite effective in
certain settings, we elicit their limitations with the following real-world
example on text-data.

\begin{example}
Consider a corpus that are questions submitted by a hotel's guests to the
concierge. Our goal is to build an intent classifier to find (and label) the
set of questions asking for directions or means of transportation from one
location to another. Below is a sample of messages from the corpus with
positive instances marked as green.

\begin{center}
\includegraphics[scale=0.32]{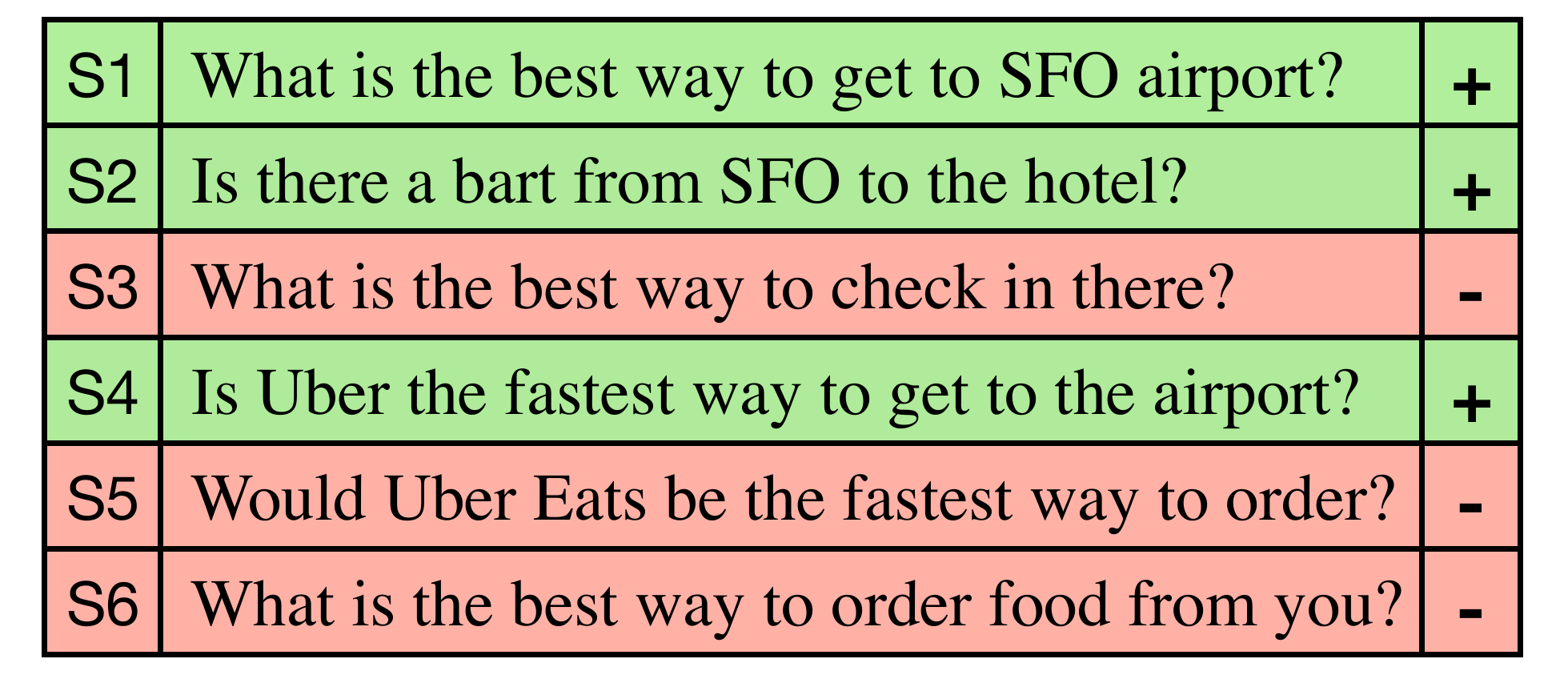}
\end{center}
\label{ex:ex1}
\end{example}

Relying on domain experts to provide labeling heuristics for tasks such as the
one presented in Example~\ref{ex:ex1} is a common approach but it has a number
of shortcomings:

\begin{itemize}[leftmargin=*]

    \item It is \textbf{time consuming}. Annotators must be familiar with the
    rule language (e.g., Stanford's Tregex or AI2's IKE language). Moreover,
    they need to be acquainted with the dataset to specify useful rules, i.e.,
    rules that label a reasonable number of instances with a small amount of
    noise. This is normally done with trial-and-error and fine-tuning of the
    rules on a sample of the corpus, which can be quite tedious.

    \item Oftentimes, some \textbf{useful rules remain undiscovered}. This is
    because annotators may miss important keywords or  possess limited domain
    knowledge. For example, the word `\emph{bart}' (which refers to a
    transportation system in California) is clearly useful for the task in
    Example~\ref{ex:ex1}. However, annotators may miss the important keyword
    `{\em bart}' or they may not even know what `{\em bart}' is (especially
    those who are not from the area).

    \item It yields rules with \textbf{overlapping results}. If multiple
    annotators work on writing rules independently, they are likely to end up
    with identical or similar rules. 
    since Hence, the number of distinct labels obtained does not always grow
    linearly with the number of annotators, which is rather inefficient.
\end{itemize}

The alternative approach would be to automatically mine useful heuristics with
systems such as Snuba or Babble labble. Both systems require a set of labeled
instances (accompanied by natural language explanations in case of Babble
labble) which can be costly and oftentimes infeasible to collect in imbalanced
settings. For instance while Example~\ref{ex:ex1} shows a balanced number of
positive and negative instances, in practice, the positive instances often make
up only a tiny fraction of the entire corpus. Hence, labeling a random sample
would not be sufficient to obtain enough positive instances. Consequently,
automatically inferring heuristic rules is not feasible using the few
discovered positive instances.

\begin{figure}[t]
    \centering
    \includegraphics[scale=0.38]{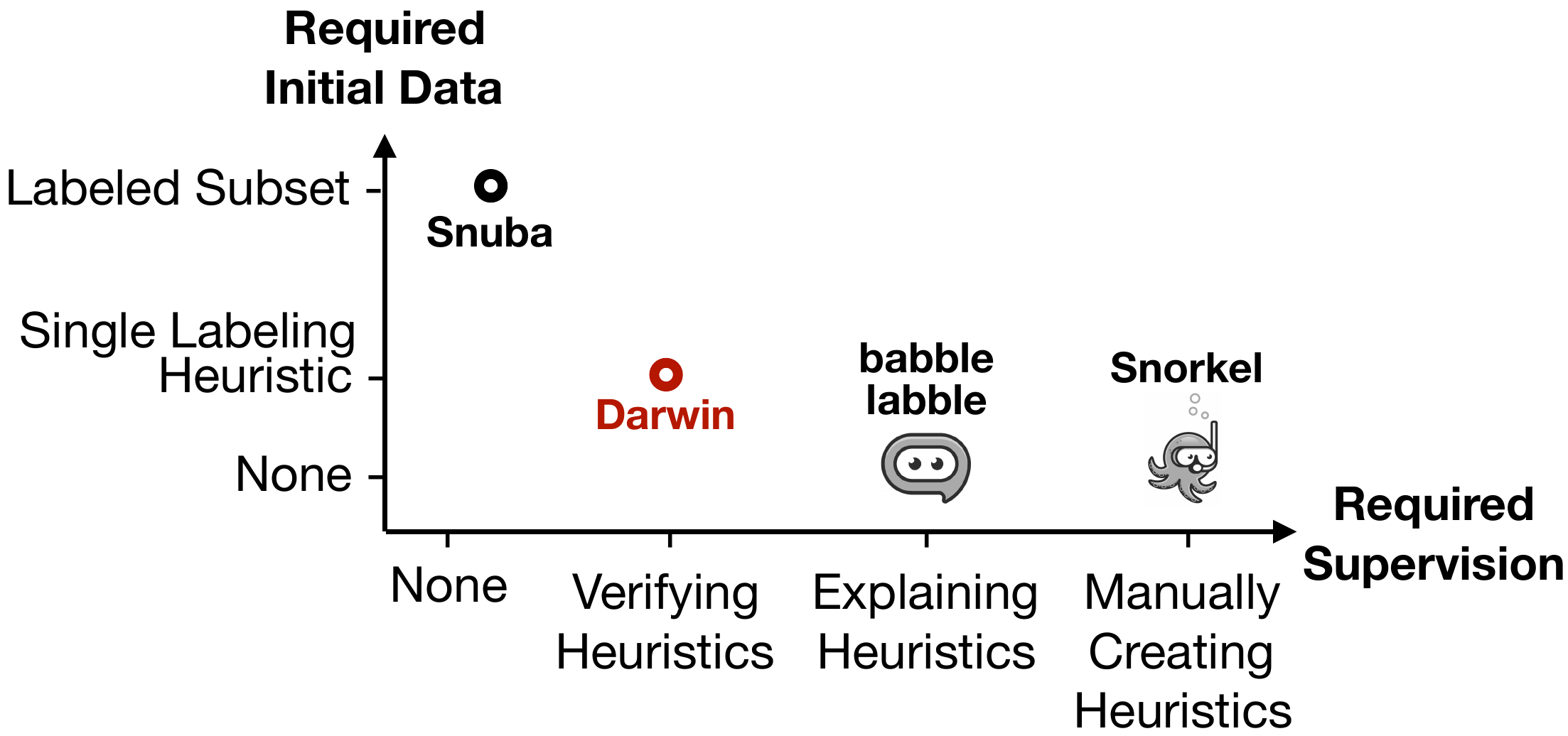}
    \caption{Comparing weak-supervision frameworks}
    \label{fig:compare}
\end{figure}

To mitigate the above issues, we present \darwin, an adaptive rule discovery
system for text data. Figure~\ref{fig:compare} highlights how \darwin\ compares
with other state-of-the-art weak-supervision frameworks. Compared to Snuba,
\darwin\ requires far less labeled instances. In fact as we show in our experiments,
a single labeling rule (or a couple of labeled instances) would be sufficient for
\darwin. Compared to Snorkel and Babble labble, \darwin\ requires a lower degree of
supervision by domain experts. More explicitly, \darwin\ requires experts to
simply verify the suggested heuristics while Snorkel requires them to manually write
such rules and Babble labble requires them to provide explanations for why a particular
label is assigned to a given data point.

Given a corpus and a seed labeling rule, \darwin\ identifies a set of promising
candidate rules. The best candidate rule (along with a few example instances
matching the rule) is then presented to the annotator to confirm whether it is
useful for capturing the positive instances or not. Figure~\ref{fig:query}
presents an example of this step for the intent described in
Example~\ref{ex:ex1}. The annotator is presented with examples that satisfy the
rule and asked to answer whether the rule is useful for the intent (a YES/NO
question). Based on the response, \darwin\ adaptively identifies the next set
of promising candidate rules. This interactive process, where rules are
illustrated with examples, facilitates annotators to identify the most
effective set of rules without the need to fully understand the corpus or the
rule language. Our contributions are as follows.

\begin{figure}
\centering
\frame{\includegraphics[scale=0.17]{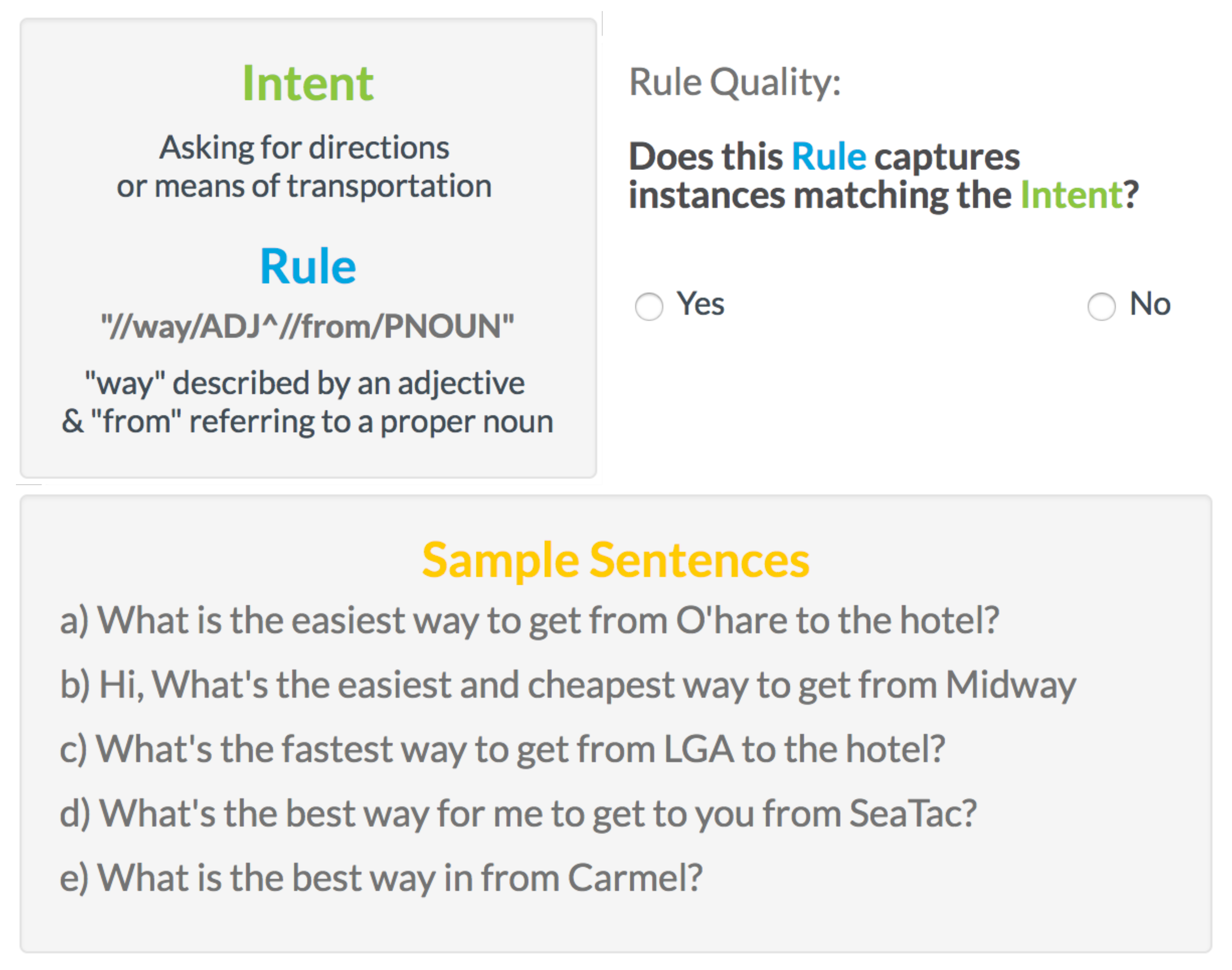}}
\caption{Sample query to annotators.\label{fig:query}}
\end{figure}

\begin{itemize}[leftmargin=*]

    \item \darwin\ supports any rule language that can be specified using a
    context-free grammar. Therefore, it can generate a wide range of rules,
    from simple phrase matching to complex conditions over the dependency parse
    trees of the sentences.

    \item \darwin\ can effectively identify rules over a large text corpora,
    even when the number of candidate rules is exponentially large. In fact,
    we show that verifying 50 heuristics suggested by \darwin\ is enough to
    achieve a F1-score of $0.80$. Furthermore, we present theoretical results on
    approximation guarantees of \darwin{}.

    \item \darwin\ does not require annotators to be familiar with the rule
    language. By analyzing the similarity and the overlap between the set of
    sentences matching different rules, \darwin\ automatically surface patterns
    in data and also supports parallel discovery of rules by asking different
    annotators to evaluate different rules.

    \item We demonstrate how \darwin\ can be used for a variety of labeling
    tasks: classify intents, find sentences that mention particular entity
    types, and identify sentences that describe certain relationships between
    entities (i.e., relation extraction).
\end{itemize}

In the following sections, we define our problem, describe \darwin, and
demonstrate its effectiveness and efficiency through a suite of experiments.
Specifically, we show that \darwin\ outperforms other baseline approaches in
its ability to generate a larger set of labeled examples by asking a limited
number of questions.

\vspace{-2mm}
\section{Preliminaries \& Problem Definition}
\label{sec:problem}
In a nutshell, \darwin\ takes as input an unlabeled corpus of sentences along
with an initial seed labeling heuristic (which is assumed to generate at least
two positive instances). \darwin\ then identifies promising candidate labeling
heuristics. \darwin\ leverages an oracle to verify whether a particular candidate
heuristic is effective at capturing positive instances or not. Finally, the 
set of discovered heuristics are forwarded to 
Snorkel~\cite{snorkel}\footnote{Note that Snorkel both provides a framework for
writing labeling rules as well as tools for training noise-aware models. Here
we refer to the latter.} to train a high precision classifier. Before describing
\darwin's rule discovery pipeline in detail, we provide a formal definition of
labeling heuristics along with a description of an oracle.

\mpara{Heuristic search space.}
Naturally, labeling heuristics can be of different types with distinct
semantics. For example, a heuristic may check for certain phrases in a
sentence~\cite{Snuba} or it may enforce some conditions on the parse
tree~\cite{koko} and POS tags of a sentence. In \darwin, the
space of possible heuristics are specified using a collection of
{\em Heuristic Grammars}, where each grammar describes a particular
type of labeling heuristics. These concepts are formally defined as follows.

\begin{definition}[Heuristic Grammar]
A Heuristic Grammar $\mathcal{G}$ is a Context-Free Grammar (CFG).
Recall that a CFG consists of a collection of derivation rules. 
\end{definition}

For a given heuristic grammar $\mathcal{G}$, we define labeling heuristics as
follows.

\begin{definition}[Labeling heuristic]
A labeling heuristic $r$ is a derivation of the grammar $\mathcal{G}$. We use $C_r$
    to denote the set of sentences in the corpus that satisfy the heuristic $r$, and 
    refer to $|C_r|$ as  it's \emph{coverage}.
\end{definition}

To further clarify the above definitions, let us consider a simplified regular
expression grammar called TokensRegex. TokensRegex captures all regular expressions
over tokens considering `+' and `*' operators\footnote{We use TokenRegex to explain 
\darwin's pipeline. \darwin\ functionality is not restricted to this grammar and we discuss more complex grammar.}. This grammar can be formally writen using a CFG grammer as shown
below.

\begin{example} [TokensRegex Grammar]
Let $\mathcal{V}$ denote the set of all possible words. The regular expression grammar on the tokens comprises of the following derivation rules.\\[1ex]
\indent $A \rightarrow vA$ \ \ \ \ $(\forall v\in\mathcal{V})$\\
\indent $A \rightarrow A+A$\\
\indent $A \rightarrow A*A$\\
\indent $A \rightarrow \epsilon$\\[1ex]
The above TokensRegex grammar allows for a regular expression of words as a candidate labeling heuristic.
For example, this grammar generates heuristics such as `best way to' or `shuttle' as well as less meaningful heuristics such as `shuttle is airport' as candidates for the task described in Example~\ref{ex:ex1}.
A sentence satisfies the heuristic if it contains that phrase. The sentences $s1$, $s3$ and $s6$ in Example~\ref{ex:ex1}
satisfy the heuristic $r = \textit{`best way to'}$, hence $C_h=\{s1,s3,s6\}$.
\label{ex:ex2}
\end{example}

As a default setting, \darwin\ comprises of two different grammars (a) TokensRegex (b) TreeMatch, with the ability to plug in more heuristic grammars as long as they are context-free. While TokensRegex is capable of capturing lexical patterns and phrases, it fails to capture syntactic patterns and pattern over parse trees. TreeMatch grammar
captures such patterns to identify more complex and generic heuristic functions.

\begin{definition}[TreeMatch Grammar] Let $\mathcal{V}$ denote the set of terminals  comprising of all the tokens and Part-of-Speech
(POS) tags~\cite{pos} present in the corpus. E.g., NOUN, VERB, etc.

\noindent\textbf{Derivation Rules:} The grammar has three fundamental operations
that make up a heuristic, namely And ($\wedge$), Child (/), and Descendant (//). 
The symbol `a/b' implies that terminal `b' should be a child of terminal `a' in
the dependency parse tree. The symbol `a//b' implies that terminal `b' should be
a descendant of terminal 'a' in the parse tree. Given that, the derivation rules of
the grammar are:
\begin{alltt}
    A \(\rightarrow\) /A
    A \(\rightarrow\) A\(\wedge\)A
    A \(\rightarrow\) //A
    A \(\rightarrow\) v (\(\forall v\in\mathcal{V}\))
\end{alltt}
\end{definition}

\begin{figure}
\includegraphics[scale=0.25]{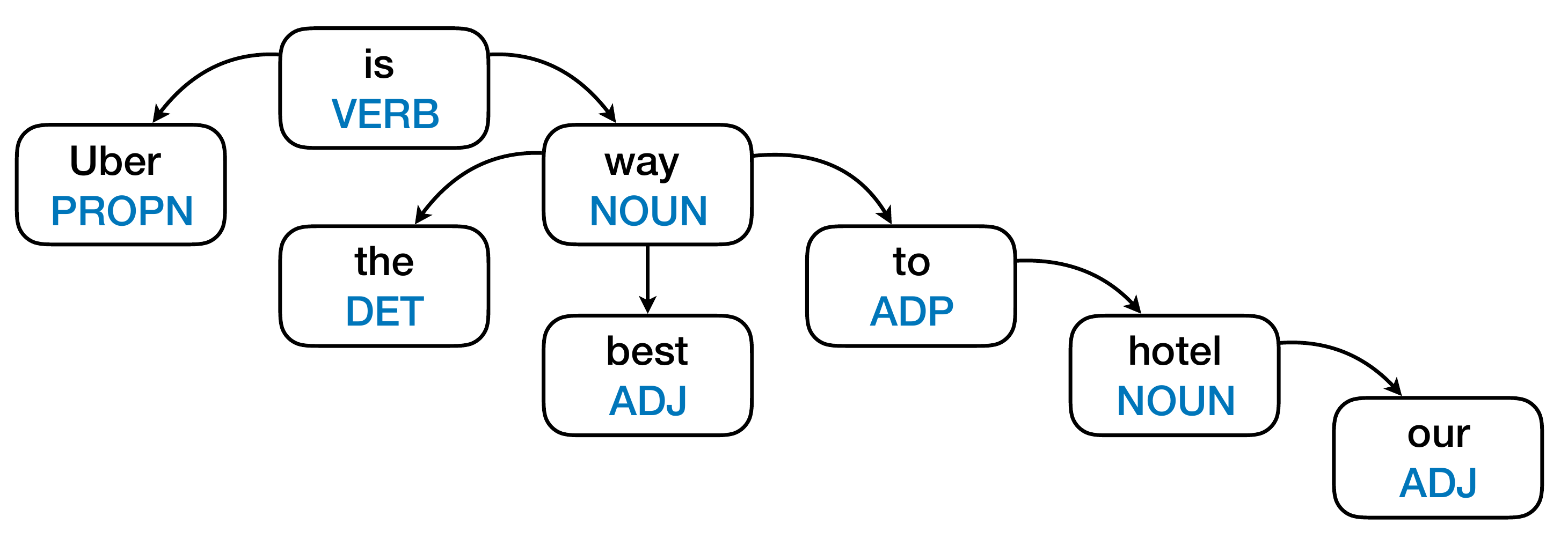}
\caption{An example of a parse tree\label{fig:tree_grammar}.}
\vspace{-.8em}
\end{figure}

It is important to mention that the complexity of heuristics that can be specified
using the TreeMatch grammar exceeds what rule-mining frameworks such as Snuba or Babble
labble can capture.



\mpara{Oracle Abstraction.} Finally, we formalize the feedback that we may either obtain from a single annotator,
a group of annotators, or a crowd-sourcing platform using the notion of Oracles as
follows.

\begin{definition}[Oracle]
An Oracle $O$ is a function which given a heuristic $r$ and a few samples from its
    coverage set $C_h$ outputs a YES/NO answer indicating whether or not $r$ is
    adequately precise. 
    \label{def:oracle}
\end{definition}
An Oracle plays the role of a perfect annotator who always answers the questions correctly.
In practice, annotators may provide incorrect answers (as we show in our experiments), but
the notion of an oracle enables us to provide insights into the theoretical aspects of our problem.

\mpara{Problem statement.}
We are now ready to formally define our problem.
Given a labeling task, our goal is to find a set $R$ of labeling heuristics such
that the union of the coverage of the heuristics in $R$,
denoted as $P = \bigcup_{r \in {R}} C_r$, would have a high recall (i.e., to contain a high ratio of the
positive instances in the corpus). 
We would like to maximize the recall of set $P$ without posing too many queries to the oracle. 
We empirically observed that users label a heuristic as precise only when the heuristic has precision at least 0.8. Hence, in this paper, we do not focus on optimizing the precision of heuristics, which we can also rely on various de-noising techniques from the weak supervision literature~\cite{snorkel}.
%

\begin{problem} [Maximize Heuristics Coverage]
Given a corpus $S$, a seed labeling function $r_0$, an oracle $O$, and a budget $b$, find a set $R$ of
    labeling heuristics using at most $b$ queries to the oracle, such that the recall of set $P$, i.e., the
    union of the coverage of heuristics in $R$, is maximized.
\label{prob:cover}
\end{problem}

\begin{lemma}
The {\sc Maximize Heuristics Coverage} problem is NP-hard.\label{lem:NPhard}
\end{lemma}
\begin{proof}
We show the hardness of our problem by reducing the maximum-coverage
problem to an instance of our problem. Given a collection of sets
$\mathcal{A} =\{A_1,A_2,\ldots,A_n\}$ and a budget $b$, the
maximum-coverage problem aims to find $b$ sets from $\mathcal{A}$
such that the size of their union is maximized. Given an instance
of the maximum-coverage problem, we create an instance of our problem
as follows. For each set $A_i$, we define a heuristic $r_i$ with coverage
set $C_{r_i} = A_i$ and mark all the instances as positives.
Consequently, the oracle $O$ would always respond with a Yes as all
the heuristic are perfectly precise. Now, it is easy to see that the
coverage of set $P$ in our setting is equivalent to the coverage
of selected sets in the maximum-coverage problem. As a result,
if our heuristic discovery problem can be solved in polynomial time,
then the corresponding sets would form the optimal maximum-coverage
solution. Hence, our problem is also NP-hard.
\end{proof}


Note that while we focus on maximizing the recall, it is also useful to report the performance of the classifier that is trained using our weakly-supervised labels.
Therefore, in our experiments, we also record the F-score of our trained classifier to provide a
better evaluation of \darwin.

\section{The {\Large \darwin} System}
\label{sec:darwin}
In this section, we describe the architecture of \darwin\ which is illustrated
in Figure~\ref{fig:darwin_arch}. \darwin{} operates in multiple phases that aim to identify diverse set of heuristics used to identify positives. The pipeline is initialized with a seed labelling function or a couple of positive sentences.  \darwin{} 
learns a rudimentary classifier using these positive sentences and  the classifier is refined with evolving training data. In order to identify new heuristics \darwin{} leverages the following properties. (i) The generalizability of the trained classifier helps guide the search towards semantically similar heuristics. For example, on identifying the importance of `bus' as a heuristic, \darwin{} identifies `public transport' as another possibility due to their related semantics\footnote{This generalization is possible via word embeddings
which are provided as an input to the classifier. We provide more details of our classifier in the experiments.}. (ii)  Local structural changes to the already identified heuristics helps identify new heuristics eg. consider `What is the best way to the hotel?' as input seed sentence, \darwin{} constructs local modifications by dropping and adding tokens (derivation rules in general)  and identifying a new heuristic `shuttle to the hotel'. 
\darwin{} leverages these intuitions to adaptively refine the search space and simultaneously learn a precise classifier with high coverage over the positives.

Before describing the architecture, we define a data structure that is critical for efficient execution of \darwin{}. All  candidate heuristics that are considered by \darwin{} are organized in the form of a hierarchy. This hierarchy captures subset/superset relationship among heuristics. Heuristics with higher coverage are placed closer to the root and the ones with lower coverage are closer to the leaves. For example. `best way to the hotel' is a subset of `best way to' and will be its descendant. One of the key properties of this hierarchical structure is that if a heuristic $r$ is identified to capture positives, then any of its subset (descendant in the hierarchy) does not capture any new positive. This data structure has $O(1)$ update time to identify the subsets of a heuristic. Additionally,  it is helpful for efficient execution of local structural changes to any heuristic. All these benefits will be discussed in detail in the later sections.

Algorithm~\ref{algo:darwin} presents the pseudo code of the end-to-end \darwin\ architecture. \darwin's input consists of the corpus to be
labeled, a collection of heuristic grammar, and one (or more) seed labeling function(s). Alternatively, a set of positive instances can be provided instead of seed labeling heuristics. 
The output of \darwin\ is the set of
generated heuristics, the positive instances that are discovered, and a classifier that is
trained using the labeled data.

\setlength{\textfloatsep}{5pt}
\begin{algorithm}[t]
\caption{\darwin\ \label{algo:darwin}}
\begin{algorithmic}[1]
{\small
\REQUIRE Input Corpus $\mathcal{S}$, seed heuristic $r_0$, budget $b$
\ENSURE 	Collection of heuristics ${R}$, Set of positive instances $P$, Classifier $C$ 
\STATE $\mathcal{I}\leftarrow \texttt{generate\_index}(S)$
\STATE $Q\leftarrow \phi$
\STATE $P\leftarrow \texttt{coverage}(r_0)$
\STATE $C,P'\leftarrow \texttt{train\_classifier}(P,\{r_0\},S)$
\WHILE{$|Q|\leq b$}
\STATE $\mathcal{H}\leftarrow \texttt{generate\_hierarchy}(\mathcal{S},P',\mathcal{I})$
\STATE $q\leftarrow$ \texttt{traversal} ($H,P,Q,C$)
\IF{\texttt{oracle\_query}(q)}
\STATE $P\leftarrow P\cup \texttt{coverage}(q)$
\STATE $C,P''\leftarrow \texttt{train\_classifier}(R,P,S)$
\ENDIF
\STATE $P'\leftarrow P''\setminus P'$
\STATE $\mathcal{H}\leftarrow \texttt{update\_scores} (\mathcal{H})$
\STATE $Q\leftarrow Q\cup \{q\}$
\ENDWHILE
\RETURN ${R,P,C}$
}
\end{algorithmic}
\end{algorithm}

Before the heuristic-discovery phase begins, \darwin\ creates an index over the input corpus for fast
access to the sentences that satisfy a given heuristic (more details in Section~\ref{sec:preprocess}).
The heuristic-discovery phase is an iterative process where \darwin\ interacts with annotators
and uses their feedback to identify new candidates and ask further queries. In a nutshell, each iteration
consists of the following steps. First, the \emph{Candidate Generation} component
generates a small set of promising candidate heuristics (from the space of all possible heuristic functions),
and organizes them in the form of a hierarchy $\mathcal{H}$ (line 6) with the most generic functions at the top and the
stricter ones at the bottom. We will describe shortly how $\mathcal{H}$ is generated and used to prune
less effective heuristics. Once the hierarchy is
built, \darwin's \emph{Hierarchy Traversal} component carefully navigates and evaluates the heuristics in
the hierarchy to find the best candidate (line 7). The best candidate is then presented to the
annotator (line 8). Finally, the updated classifier and scores of 
heuristics are sent back to hierarchy generation to identify new candidates and perform traversal for the next iteration.
We describe the details of these components next.

\begin{figure}
\includegraphics[scale=0.23]{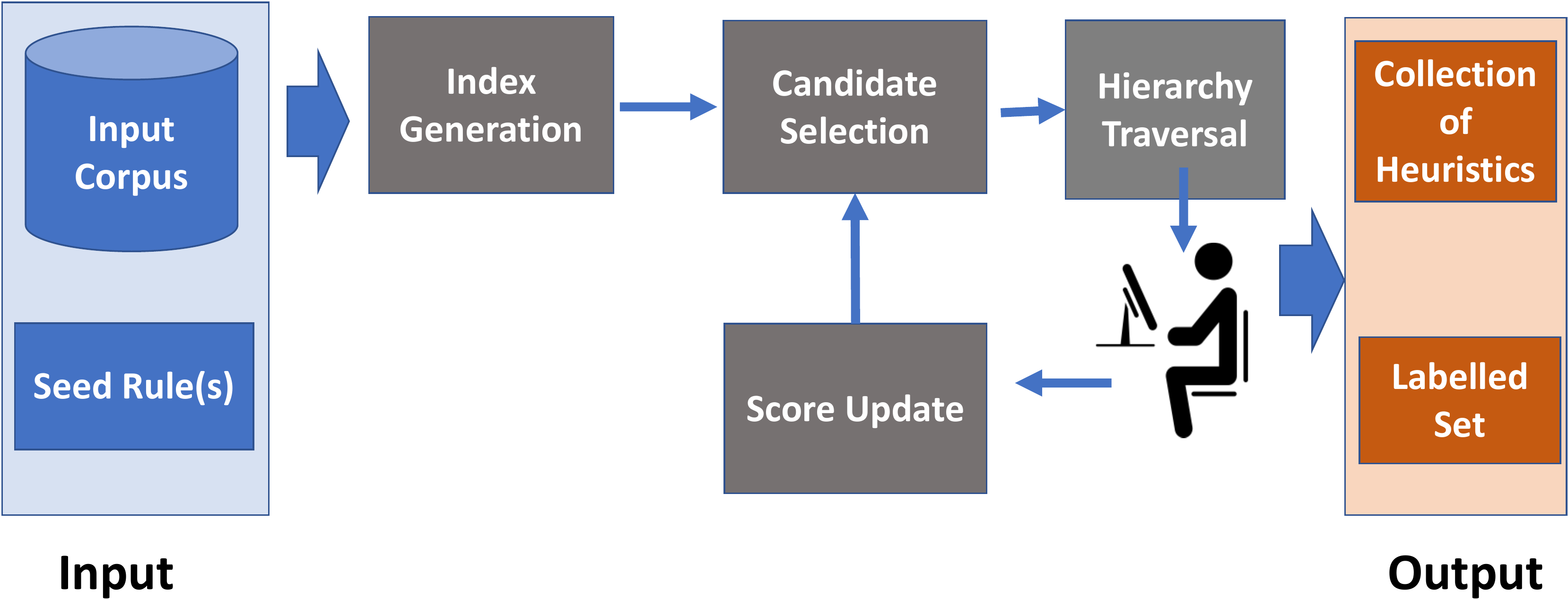}
\caption{\darwin{}'s architecture\label{fig:darwin_arch}. }
\vspace{-1.3em}
\end{figure}

\subsection{Indexing the Input Corpus\label{sec:preprocess}}

\darwin\ creates an index for the input corpus to provide fast access to sentences that
satisfy certain heuristics. This index aims at constructing a space efficient representation of each sentence in the corpus for efficient execution of subsequent steps involving traversal through the various candidate heuristics. This hierarchical structure of this index is very similar to that of a trie.

Given a collection of heuristic grammar $\{G_1,\ldots, G_t\}$ and a sentence $s$, one can enumerate the set
of all possible heuristics of ${G_i}$, generated using  a fixed number of
derivation rules, that $s$ satisfies.
For example, using the TokensRegex grammar, the set of all heuristics
that a sentence $s$ satisfies is the set of all regular expressions that correspond to $s$. We
organize the set of heuristics matching a sentence $s$ into a structure called the
\emph{Derivation Sketch},
which 
summarizes the derivations of all heuristics that match $s$.
Figure~\ref{fig:op_sketch} shows (parts of) the derivation sketch for sentences
$s1$ and $s4$ from Example~\ref{ex:ex1} based on the TokensRegex grammar. 

\begin{figure}
\includegraphics[scale=0.40]{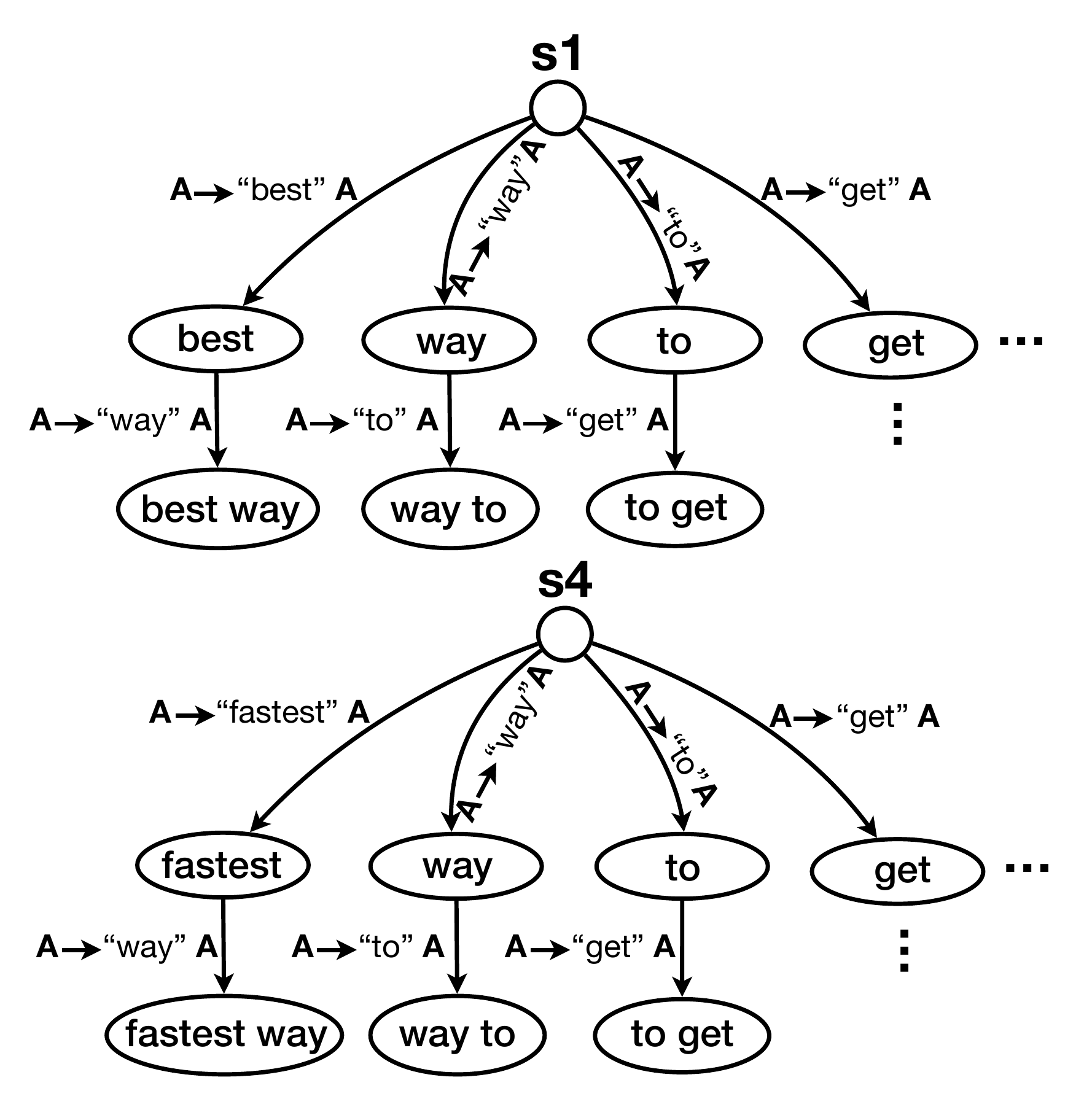}
\vspace{-.2em}
\caption{Examples of derivation sketches\label{fig:op_sketch}}
\vspace{-1em}
\end{figure}

A derivation sketch is built for each sentence in the corpus. After this, we create an
index $\mathcal{I}$ which is a compact representation of all heuristics
that are satisfied by at least one sentence in the corpus. Each node in $\mathcal{I}$
represents a heuristic labeling function and stores the number of sentences that satisfy it,
pointers to the children in the index, and an inverted list that points to
sentences that satisfy the heuristic.

The index $\mathcal{I}$ is created by merging the derivation sketch of sentences, one at a time.
The index is first initialized with the derivation sketch of
the first sentence. Thereafter, for every new sentence $s$, the derivation sketch of $s$ is
merged into $\mathcal{I}$ as follows.
The root node of the sketch
and the root node of $\mathcal{I}$ is merged, and then all nodes (starting form
the merged root) are considered in a breadth-first fashion; The children of the
node under consideration which are derived using the same derivation rule are merged
together. For every node that gets merged, the count of the merged node
is increased by one. Also, the inverted list at that node gets updated to
include the new sentence. Figure~\ref{fig:genindex}
shows the index built from derivation sketches of $s1$ and $s4$ in Figure~\ref{fig:op_sketch}.
Note that the time taken to construct the index is linear in the number of sentences since
we have limited the number of derivation rules possible to generate a heuristic. 
The process is also highly parallelizable as index structures for different parts of the corpus can
be created independently and
then merged. This index also has a linear update time complexity for adding the derivation
sketch of a new sentence.

\mpara{TreeMatch Grammar:} 
This grammar has more operators as compared to TokensRegex and can generate exponentially more candidate heuristics. The derivation sketch can be created as explained by enumerating all sequence of
derivation rules up to a fixed number of steps. However, a more compact derivation sketch for
the TreeMatch grammar is simply the dependency parse tree of the sentence, as we can use
it to quickly check whether a heuristic matches the parse tree or not \cite{koko}.
Figure \ref{fig:tree_grammar} shows the dependency parse tree of a sentence which can serve as its derivation sketch as well.
Given the exponentially many candidate heuristics generated by this grammar, the candidate
generation step is crucial for ignoring useless heuristics and thereby helping the subsequent stages to focus on meaningful heuristics. We evaluate the performance of \darwin\ with this grammar in the
next section.
\begin{figure}
\includegraphics[scale=0.33]{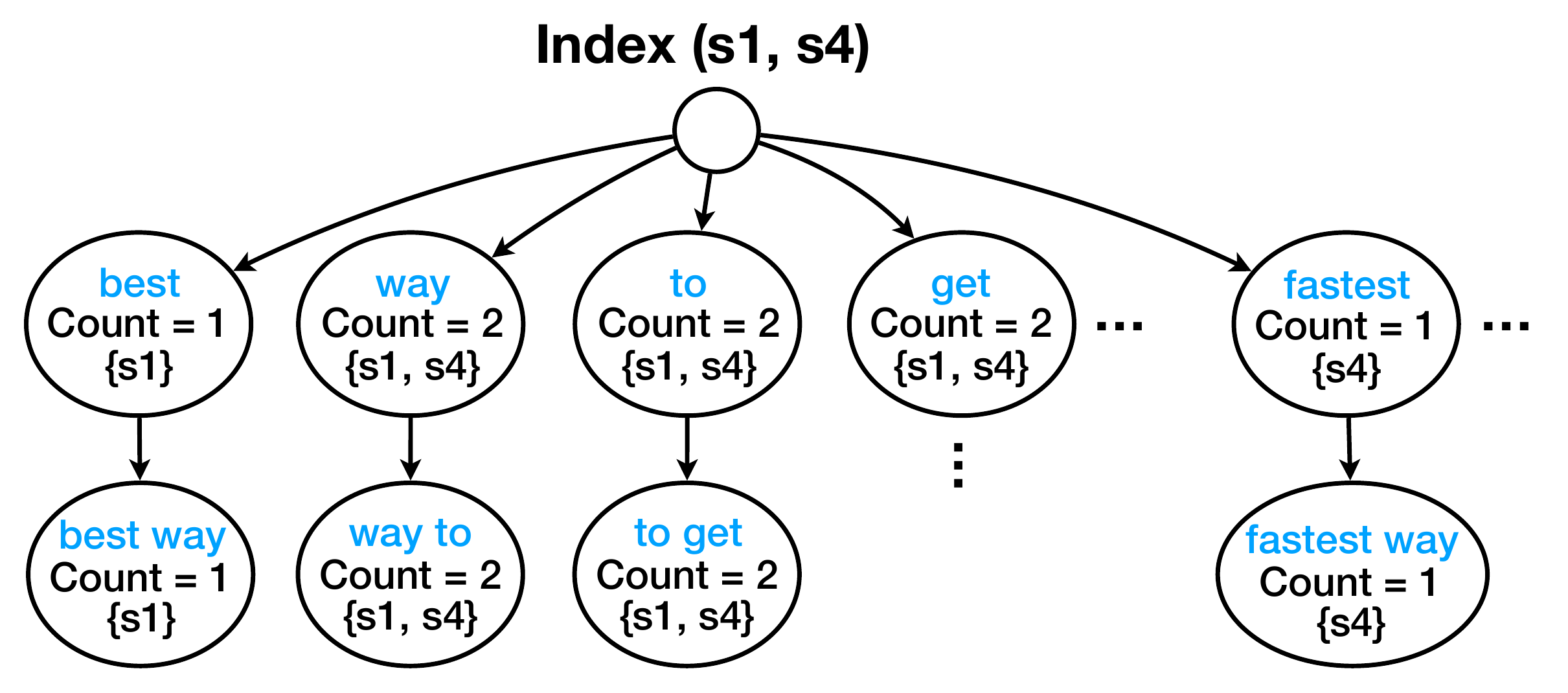}
\caption{An example of index creation process\label{fig:genindex}}
\end{figure}

\subsection{Heuristic-Hierarchy Generation\label{sec:rulegen}}
As mentioned earlier, the number of possible heuristics under a given grammar $\mathcal{G}$
is often exponential in the size of dictionary. The task of the heuristics-hierarchy generation component is to 
generate a manageable set of promising candidate heuristics from the space of all possible heuristics
and organize the generated candidate heuristics in a hierarchy that captures the subset/superset
relationship between the heuristics. Specifically, the hierarchy generation process consists
of the following steps.
First, the \emph{Candidate Generation} step generates a subset of possible heuristics that
have high coverage over the set of positive instances discovered so
far. This algorithm operates in a greedy best-first search mechanism to identify valuable candidates. These heuristics are promising as they already have some overlap with
the existing positive instances. Next, these candidates are arranged in the form of a hierarchy along with subset-superset edges between them.  We describe these steps in detail
next.

\subsubsection{Candidate Generation}
The candidate generation step uses the index $\mathcal{I}$ to generate a set of
heuristic labeling functions with high coverage over the set of positive instances $P$ that are discovered so
far by \darwin. Note that heuristics that (at least partially) cover the set of discovered
positive instances, are likely to be good heuristics and help detect more positive instances.
To efficiently find such heuristics, we rely on one of the interesting properties of index $\mathcal{I}$;
Recall that the count of a node $u\in \mathcal{I}$ refers to the total number of
sentences that contain the tokens along the path from the root to $u$ in their
derivation sketches. As descendant nodes correspond to stricter heuristics, 
the coverage of a heuristic corresponding to a node is never less than the
coverage of any of its descendants in the index. Thus we use a greedy
algorithm to identify a collection of diverse heuristics that have high coverage
over the set of positive instances $P$.

\setlength{\textfloatsep}{5pt}
\begin{algorithm}[t]
\caption{Candidate-heuristic Generation \label{algo:cand}}
\begin{algorithmic}[1]
{\small
\REQUIRE Index $\mathcal{I}$, Set of positive instances $P$, Number of desired heuristics $k$
\ENSURE 	Collection of heuristics ${R}$ 
\STATE ${R}\leftarrow \{*\}$, \ \  $\textit{recentHeuristic}\leftarrow *$, \ \ $\textit{candidates} \leftarrow \phi$
\WHILE {$|{R}|\leq k$}
\STATE $\textit{candidates}\leftarrow \textit{candidates\ }\cup \text{Children}(\textit{recentHeuristic} ,\mathcal{I})$
\STATE $\textit{sortedCandidates} \leftarrow \text{CoverageSort}(\textit{candidates}, P)$
\STATE $\textit{recentHeuristic} \leftarrow \textit{sortedCandidates}[0]$
\STATE $\textit{candidates} \leftarrow \textit{candidates}.\text{remove}(\textit{recentHeuristic})$
\STATE ${R} \leftarrow {R}\cup \textit{recentHeuristic}$
\ENDWHILE
\RETURN ${R}$
}
\end{algorithmic}
\end{algorithm}

Algorithm~\ref{algo:cand} generates candidate heuristics  by exploiting 
the property described above.
The set of candidate heuristics is initialized with heuristic `*' which refers
to the root of index $\mathcal{I}$. This heuristic 
matches all possible sentences in the
corpus. In each iteration, the algorithm adds the children of the previous
iteration's best candidate heuristic to the candidate list (line 3). The candidates
are then sorted in decreasing order of coverage over the set $P$ (line 4). The
candidate with the highest coverage is
removed from the candidate list and appended to list of final results
${R}$ (lines 6-7). This process is repeated until there are $k$ heuristics in ${R}$.
Note that the time complexity of this greedy algorithm is linear in the
number of candidates generated.


Other constraints can also be added to the candidate-heuristic generation phase to ensure
that the generated heuristics satisfy those criteria. For example, \darwin\ can apply heuristics
to ensure that the candidate heuristics are diverse in terms of the set of derivation rules used to
derive the heuristic, their level in the index $\mathcal{I}$, and the set of instances they cover. 
Some of these heuristics help \darwin\ avoid having to evaluate many similar candidate heuristics.

\mpara{Hierarchical Arrangement and edge discovery.} The candidates returned by Algorithm~\ref{algo:cand} have high coverage over the positives (discovered so far).  This component iterates over the generated heuristics to arrange  them into a hierarchy $\mathcal{H}$ following the
same parent/child relationship that index $\mathcal{I}$ captures\footnote{heuristic $r_1$ is a child of
$r_2$ if it can be obtained by applying a derivation rule to $r_2$.} and an edge is added between them.

This hierarchical arrangement of heuristics is followed by a cleanup to get rid of heuristics that do not add any new positive sentences than the ones already identified. The goal of cleanup is to improve the efficiency and space complexity of \darwin\ as the traversal component will never query a heuristic that does not add any new positives.

\subsection{Hierarchy Traversal}
The result of the heuristic-hierarchy generation is a hierarchy $\mathcal{H}$ of promising heuristics. The
hierarchy traversal module determines which heuristic in the hierarchy is the best heuristic to be submitted
to the oracle.

We present three hierarchy traversal techniques:
{\babystep}, \\{\universal}, and \hybrid{}. At a high level, \babystep\ relies on the hierarchy
structure to select the next best candidate from the immediate neighborhood of heuristics
verified by the oracle in the past.
In contrast, \universal\ ignores locality constraints and selects the heuristic with maximum benefit globally. 

Finally, the \hybrid\ traversal combines the
first two techniques 
to find the next best heuristic. The \hybrid\ traversal is more robust than \babystep\ and \universal{}.  All three
techniques work in an iterative fashion, and in each iteration, the criteria
for selecting a heuristic to be sent to the oracle is
based on how {\em beneficial} the heuristic is, which we elaborate next.

\mpara{Benefit of a heuristic (r):}
The benefit of a heuristic $r$ is the expected gain in the positive
set $P$ upon choosing $r$. More formally, the benefit is quantified as
$\sum_{s\in C_r\setminus P} p_s$, where $p_s$ is the probability of sentence
$s$ being a positive instance. In \darwin, these probability values are
estimated by training a classifier\footnote{Any short text classifier would be
ideal for this task.} using the set of positive instances discovered so far and
sampling random instances from the corpus as negatives. The probability
estimates improve as the system iteratively discovers more  heuristics and
the classifier is re-trained with more positive training examples.

We describe our three traversal techniques next.

\subsection{\large \babystep}
\babystep\ traversal algorithm (Algorithm~\ref{algo:babystep}) benefits from
the local hierarchy structure around the heuristics already identified as useful
by the oracle to identify the next best heuristic for querying. Specifically,
\babystep\ maintains a set of candidate heuristics, and selects the most beneficial 
heuristic $r$ from the candidates. If the oracle confirms that $r$ is adequately
precise, then it adds $r$'s parents into the candidate set as they are
generalizations of $r$ and might be helpful at capturing more positive
instances. However, if the oracle labels $r$ as a noisy heuristic, \babystep\ adds
the children of $r$ to the candidate set instead with the hope that a specialized
version of heuristic $r$ might be less noisy.

\babystep\ is simple and efficient at utilizing the structure of the
hierarchy to find promising heuristics to submit to the oracle. Since the algorithm
only explores the local neighborhood of the queried candidates, it has a time
complexity of O($dt$), where $d$ is the maximum degree of an internal node and
$t$ is the number of iterations the algorithm is running for. However, a
disadvantage of \babystep\ is that it may require many traversal steps in cases
where the initial seed heuristic is quite different from other precise heuristics the
system aims to discover. Also, it does not exploit the similarity and the overlap
between the coverage sets of different heuristics. The \universal{} algorithm, which
we describe shortly, addresses these shortcomings by utilizing a holistic view
of the hierarchy.

\setlength{\textfloatsep}{5pt}
\begin{algorithm}[t]
\caption{\babystep{} Traversal\label{algo:babystep}}
\begin{algorithmic}[1]
{\small
\REQUIRE heuristic hierarchy $\mathcal{H}$, Seed heuristic $r_0$
\ENSURE 	Collection of positive instances $P$, Collection of heuristics $R$
\STATE {$QueryCount\leftarrow 0$}
\STATE $R\leftarrow \{r_0\}$,\ \ $P\leftarrow C_{r_0}$, \ \ $\mathcal{C}\leftarrow \text{TrainClassifier} (P)$
\STATE $\textit{localCandidates} \leftarrow \{r_0\}$
\WHILE {{$QueryCount<b$}}
\STATE $r\leftarrow \text{GetMostBeneficialCandidateheuristic}(\textit{localCandidates},\mathcal{C})$
\STATE {$QueryCount\leftarrow QueryCount+1$}
\IF{ OracleResponse$(r)$ is YES}
\STATE $R\leftarrow R\cup r$, \ \ $P\leftarrow P\cup C_r$
\STATE $\textit{localCandidates} \leftarrow (\textit{localCandidates} \setminus\{r\}) \cup \text{Parents} (r)$
\STATE $\mathcal{C}\leftarrow \text{TrainClassifier} (P) $
\ELSE
\STATE  $\textit{localCandidates} \leftarrow (\textit{localCandidates} \setminus\{r\}) \cup \text{Children} (r)$
\ENDIF
\ENDWHILE
\RETURN $P, R$
}
\end{algorithmic}
\end{algorithm}

\mpara{Efficient Implementation.} Since the \babystep\ traversal only
explores a node's immediate parents/children, it does not require the entire hierarchy
apriori. Hence, in its implementation, we can skip the heuristic-generation component and 
expand the hierarchy on the fly based on the oracle's feedback.

\vspace{-1.5mm}

\subsection{\large \universal}
The \universal\ algorithm (see Algorithm \ref{algo:universal}) evaluates all
heuristics present in the hierarchy to identify the best heuristic. In each iteration,
\universal\ omits any heuristic for which the benefit per instance is smaller than
$0.5$, i.e. majority of the instances in $C_r$ are expected to be negatives.
Among the remaining heuristics, it chooses the heuristic with maximum benefit to submit to
the oracle. Based on oracle's feedback, it re-trains the classifier if new
positives were discovered or else it continues to query the next best heuristic to
the oracle. Note that \universal\ captures the best candidates irrespective of
the hierarchy structure.

The strength of \universal\ is in its capability to identify semantic similarity
between heuristics and their matching instances even if they are far apart
in the hierarchy. However, it has the following shortcomings: (1) compared to
\babystep, it is inefficient as it iterates over all heuristics in the hierarchy to
identify the best candidate, and (2) in absence of enough positive instances,
the trained classifier is likely to overfit and not generalize well to other
precise heuristics. In such cases, \universal\ fails to exploit the structure of the
hierarchy to at least find heuristics that are structurally similar to the seed heuristics.

We describe the \hybrid\ algorithm next, which combines the strengths of 
\universal\ and \babystep.

\setlength{\textfloatsep}{5pt}
\begin{algorithm}[t]
\caption{\universal{}  Traversal\label{algo:universal}}
\begin{algorithmic}[1]
{\small
\REQUIRE heuristic hierarchy $\mathcal{H}$, Seed heuristic $r_0$
\ENSURE 	Collection of positive instances $P$, Collection of heuristics $R$
\STATE {$QueryCount\leftarrow 0$}
\STATE $R\leftarrow \{r_0\}$,\ \ $P\leftarrow C_{r_0}$, \ \ $\mathcal{C}\leftarrow \text{TrainClassifier} (P)$
\STATE $\textit{universalCandidates} \leftarrow \{r : r \in \mathcal{H}\}$
\STATE $\mathcal{C}\leftarrow \text{TrainClassifier}(P) $
\WHILE {{$QueryCount<b$}}
\STATE $r\leftarrow \text{GetMostBeneficialCandidate}(\textit{universalCandidates},\mathcal{C})$
\STATE {$QueryCount\leftarrow QueryCount+1$}
\STATE $\textbf{if} \text{\ AvgBenefit}(r) \leq 0.5 \textbf{\ then} \text{\ continue}$
\IF {$\text{OracleResponse}(r)\text{\ is YES}$}
\STATE $R\leftarrow \{r_0\}$,\ \ $P\leftarrow C_{r_0}$
\STATE $\mathcal{C}\leftarrow \text{TrainClassifier}(P) $
\ENDIF
\STATE $\textit{universalCandidates} \leftarrow \textit{universalCandidates} \setminus \{r\}$
\ENDWHILE
\RETURN $P, R$
}
\end{algorithmic}
\end{algorithm}

\vspace{-1.5mm}

\subsection{\large \hybrid{}}

\hybrid\ (See Algorithm \ref{algo:hybrid}) combines the two previous traversal techniques
by maintaining a list of local candidates and a list of universal candidates, and imitating
the strategy of the both traversal algorithms. Starting from the \universal\ strategy,
the \hybrid\ algorithm queries candidate heuristics (with a benefit per instance above $0.5$) to
the oracle. If the algorithm fails to find a precise heuristic within a fixed number of attempts, then it switches to the \babystep\ strategy. Similarly, if the \babystep\ strategy has
no success within a fixed number of attempts, the traversal toggles to the \universal\ strategy. 
The switch between the two strategies is decided based by a parameter $\tau$
(by default 5) which denotes the number of unsuccessful attempt before the switch happens.
Clearly, higher values of $\tau$ discourage switching between the two
strategies.

\setlength{\textfloatsep}{5pt}
\begin{algorithm}[t]
\caption{\hybrid{}  Traversal\label{algo:hybrid}}
\begin{algorithmic}[1]
{\small
\REQUIRE heuristic hierarchy $\mathcal{H}$, Seed heuristic $r_0$
\ENSURE 	Collection of positive instances $P$, Collection of heuristics $R$
\STATE $\textit{universalMode}\leftarrow \text{True}$,\ \  $\textit{attempt}\leftarrow 0 $
\STATE $R\leftarrow \{r_0\}$,\ \ $P\leftarrow C_{r_0}$, \ \ $\mathcal{C}\leftarrow \text{TrainClassifier} (P)$
\STATE $ \textit{localCands}\leftarrow \{r_0\}$, \ \ $\textit{universalCands} \leftarrow \{r : r \in \mathcal{H}\}$
\STATE {$QueryCount\leftarrow 0$}
\WHILE {{$QueryCount<k$}}
\IF{$\textit{attempt} \geq \tau$}
\STATE $\textit{universalMode}\leftarrow \textbf{not } \textit{universalMode}$
\STATE $\textit{attempt}\leftarrow \textit{0}$
\ENDIF
\STATE $\textit{attempt}\leftarrow \textit{attempt}+1$
\STATE $\textit{candidates} = \textit{universalCands} \textbf{\ if} \textit{\ universalMode} \textbf{\ else} \textit{\ localCands}$
\STATE {$QueryCount\leftarrow QueryCount+1$}
\STATE $r\leftarrow \text{GetMostBeneficialCandidateheuristic}(\textit{candidates},\mathcal{C})$
\STATE $\textbf{if} \textit{\ universalMode} \text{\ and AvgBenefit}(r) \leq 0.5 \textbf{\ then} \text{\ continue}$
\IF{$\text{OracleResponse(r)}$ is YES}
\STATE $R\leftarrow R\cup r$, \ \ $P\leftarrow P\cup C_r$
\STATE $\mathcal{C}\leftarrow \text{TrainClassifier}(P)$
\STATE $\textit{localCands} \leftarrow \textit{localCands} \setminus \{r\} \cup \text{Parents}(r)$
\ELSE
\STATE $\textit{localCandidates}\leftarrow \textit{localCandidates} \setminus \{r\} \cup \text{Children}(r)$
\ENDIF
\STATE $\textit{universalCands} \leftarrow \textit{universalCands} \setminus \{r\}$
\ENDWHILE
\RETURN $P, R$
}
\end{algorithmic}
\end{algorithm}

Our empirical evaluation shows that \hybrid\ formed by combining \universal\ and \babystep\  strategies, runs well on all
types of datasets even when the other two traversal algorithms struggle to
discover high-quality heuristics. In short, if the trained classifier is noisy (due
to lack of positive instances), \hybrid\ exploits the structure of the hierarchy
to search for precise heuristics. Similarly when no precise heuristics are found by
\babystep, it uses \universal's ability to generalize to other heuristics. 


\subsection{Score Update}
After a query is submitted to the oracle, \darwin\ passes the feedback to the \emph{score update} component to (1) re-train the classifier, (2)
re-evaluate the scores of all heuristics in the hierarchy, and (3) update the set of
positive instances (if feedback is positive) and signal the hierarchy generation
component to generate new candidate heuristics to be added to the hierarchy.

\subsection{Theoretical Analysis}
In this section we analyze the ability of \universal{} hierarchy traversal in
identifying positive sentences within a query budget $b$. For this analysis,
we consider a simple model capturing how positive and negative instances 
are scored by a classifier. An ideal classifier assigns a score of $1$
to positives and $0$ to negatives. However in practice, the scores follow
a different distribution which we model as follows. Let $P^*$
denote the collection of positive sentences in the corpus of sentences $S$.
We assume that a reasonable classifier assigns to a positive sentence $s\in P^*$ a
score larger than $\theta\ge 0.5$ with probability $\beta$, and less than $1-\theta$
otherwise. Similarly, the score of a negative sentence $s\in S\setminus P^*$ is above
$\theta$ with probability $\beta'$. Naturally, for a classifier is better than random
$\beta$ is larger then $\beta'$. Under this model, we can show that the set of
heuristics $R$ and the corresponding positive set $P$ identified by \universal\ are
constant approximation of the optimal solution. 

We also make a few assumptions about the hierarchy of heuristics. We assume that
the number of heuristics in the hierarchy $\mathcal{H}$ is linear to the number of
sentences (i.e., $O(n)$) and each heuristic has a minimum coverage of $\Omega(\log n)$.
This is a realistic assumption as we focus on heuristics that can be derived from their
context free grammar using a fixed number of steps, and our algorithm is aimed at
identifying heuristics that cover a large fraction of positives. Under this assumption,
we show that at any iteration, a heuristic $r$ chosen by \universal{} has coverage
$|C_r|$ larger than  $\alpha \times \max_{r'\in \mathcal{H}} |C_{r'}|$,
where $\alpha$ is a constant. This guarantees that \universal{} identifies at least
$\alpha P_{OPT}$ positives within a query budget of $b$, where $P_{OPT}$ is the total number of positives identified by an
ideal algorithm. To bound the estimated coverage of a heuristic $r$, we use the
Hoeffding's inequality~\cite{hoeffding1994probability}.

\spara{Notation.} We define a random variable $X_s$ which refers to the score assigned to a sentence $s$ and let $\mu_s$ denote its expected value. The benefit score of a heuristic $r$ is $\sum_{s\in C} X_s$ and its expected value is denoted by $\mu_r$.

\begin{lemma}
Given a heuristic function $r$ with coverage of $C_r$ and precision $p$, the expected score of the heuristic function is at least $\theta \beta' |C_r| $.\label{lem:exp1}
\end{lemma}
\begin{proof}
Expected score of the heuristic function is 
\begin{eqnarray*}
E\left[\sum_{s\in C_r} X_s\right] &=& \sum_{s\in C_r} \mu_s= \sum_{s\in C_r\cap P^*} \mu_s +\sum_{s\in C_r\setminus P^*} \mu_r  \\
&\ge& \sum_{s\in C_r\cap P^*} \left(\theta \beta\right) +\sum_{s\in C_r\setminus P^*} \left( \theta \beta'\right)   \\
&=&  \left(\theta \beta\right)p|C_r|  + (1-p)|C_r|\left( \theta \beta'\right)   \\
&\ge&\theta \beta' |C_r|
\end{eqnarray*}
\end{proof}
We use this calculation to bound the score of a heuristic $r$ that has more than $\log n$ sentences.
\begin{lemma}
Consider a heuristic function $r$  with coverage $C_r$ such that  $|C_r|=c\log n$ sentences, where $c\geq \frac{2}{\epsilon^2\theta^2\beta'^2}$ is a constant. The benefit score of the heuristic is at least  $ (1-\epsilon) \theta \beta' |C_r|$ with a probability of $1-2/n^4$.\label{lem:s1}
\end{lemma}

\begin{proof}
The score of heuristic function $r$ is $\sum_{s\in C_r} X_s$. The expected value of the score (denoted by $\mu_r$) is calculated in lemma \ref{lem:exp1}. Using Hoeffding's inequality,
\begin{eqnarray*}
Pr\left[\frac{1}{|C_r|}\sum_{s\in C_r} X_s \leq  (1-\epsilon)\mu_r/|C_r| \right] &\leq& 2e^{-2\epsilon^2\mu_r^2/|C_r|}\\
&\le& 2e^{-2\epsilon^2\theta^2\beta'^2|C_r|}\\
&\le& 2e^{-4\log n}=\frac{2}{n^4}\\
\end{eqnarray*}
This shows that $\frac{1}{|C_r|}\sum_{s\in C_r} X_s$ is greater than $(1-\epsilon) \theta \beta'|C_r|$ with a probability more than $1-\frac{2}{n^4}$
\end{proof}
Using a similar analysis, we identify an upper bound of the heuristic score. Due to space constraint, we defer the proof to Appendix.

\begin{lemma}
Given a heuristic function $r$ with a coverage of $C_r$ and precision $p$, the expected score of the heuristic function is atmost $ \left( \beta+(1-\theta)(1-\beta)\right)|C_r|$.\label{lem:exp2}
\end{lemma}

\begin{lemma}
Consider a heuristic $r$  with coverage $C_r$ such that  $|C_r|=c\log n$ sentences, where $c\geq \frac{2}{\epsilon^2(\beta+(1-\theta)(1-\beta))^2}$ is a constant. The score of the heuristic is atmost  $ (1+\epsilon) \left( \beta+(1-\theta)(1-\beta)\right)|C_r|$ with a probability of $1-2/n^4$ \label{lem:s2}
\end{lemma}

Using the calculated bounds on score of a heuristic, we evaluate the condition when a particular heuristic is preferred over the other.

\begin{lemma}
Given a pair of heuristic functions $r_1$ and $r_2$ with respective coverage $C_1$ and $C_2$. If $C_1$ has more positives than $C_2$, the \universal{} score of $r_1$ is higher than that of $r_2$ whenever $\frac{|C_1|}{|C_2|} \geq \alpha $ with a probability of $1-\frac{4}{n^4}$, where $\alpha$ is a constant.\label{lem:main}
\end{lemma}

Using a similar analysis, we can calculate the estimated average probability of a heuristic. For a heuristic $r$ with precision $p_r$, we can show that it is considered for benefit calculation only when $p_r> \gamma$, where $\gamma$ is a constant.

\begin{theorem}
In worst case, \universal\ provides a constant approximation of Problem \ref{prob:cover} with a probability of $1-o(1)$
\end{theorem}
\begin{proof}
In each iteration, \universal{} algorithm sorts each of the candidate heuristic based on estimated average probability of a randomly chosen sentence from $C_r$ to be positive. All these candidates have true precision $p_r>\gamma$. Given a pair of heuristics $r_1$ and $r_2$, using Lemma~\ref{lem:main}, the benefit score of a block $r_1$ is higher than that of $r_2$ whenever $\frac{|C_{r_1}|}{|C_{r_2}|} > \alpha$ with a probability of $1-\frac{4}{n^4}$. Let $r_{OPT}$ be the heuristic chosen by optimal algorithm. Using union bound over ${n\choose 2}$pairs of heuristics, with a probability of at least $ 1-\frac{4}{n^2}$, the estimated benefit of $r_{OPT}$ is higher than that of any $r'$ whenever $|C_{r'}|\leq |C_{r_{OPT}}|\alpha$. Therefore, \universal\  never chooses any block heuristic with coverage smaller than $|C_{r_{OPT}}|\alpha$ with a probability of $1-o(1)$.
This shows that the total number of positives identified by \universal are at least $|C_{r_{OPT}}|\alpha \gamma$, which is a constant approximation of $|C_{r_{OPT}}|$ with a probability of $1-o(1)$.
\end{proof}

Notice that analysis above makes certain assumptions about the quality of the classifier.
In the initial iterations of \darwin\ the classifier has a low recall and hence, the values of $\beta$ and $\theta$ are lower. As \darwin\ identifies new heuristic functions, the increase in training data pushes these values higher, thereby improving the approximation factor of our algorithm. It is important to mention that even when the classifier is not ideal, our only key
assumption is that the classifier would perform better than random.

\mpara{Discussion.} We proposed three techniques for hierarchy traversal. \universal{} approach is useful to capture holistic information about the different candidate labeling heuristics and is proven to achieve constant approximation of the optimal solution under reasonable assumptions of the trained classifier. However, due to lack of training data in initial iterations of the pipeline, this assumption may not hold and \universal{} does not perform optimally. However, \babystep{}  performs local generalization of identified heuristics to quickly increase the number of identified positives. \hybrid{} is a robust amalgamation of these two techniques and is recommendation. Since \hybrid{} is a generalization of \babystep{} and \universal{}, it is slightly less efficient than either of these. 


\vspace{-2mm}
\section{Experiments}
\label{sec:experiments}
In this section, we perform empirical evaluation of \darwin{} along with other baselines to validate the following.
\begin{compactitem}
    \item The ability of \darwin\ to identify majority of the positives even when initialized with a small seed set.
    \item The  positives identified by \darwin\ outperform other baseline techniques that use active learning, a human annotator or any other automated techniques. We show that \darwin\ can uncover
most of positive instances (i.e, 80\% or more) with roughly 100 queries.
    \item The heuristics identified by \darwin\ have high precision and help train a classifier with superior F-score ($\geq 0.8$).
    \item \darwin\ is highly efficient and can generate labels from a corpus of 1M sentences in less than 3 hrs. \darwin\ performance is resilient to variations in the seed set.
\end{compactitem}

\subsection{Experimental Setup}
Here, we describe the datasets, the baselines, and our overall experimental setup.

\begin{table}[t]
\small
\begin{center}
\begin{tabular}{||c | c | c |c||} 
\hline
\textbf{dataset} & \# Sentences & \% Positives & Labeling  \\ [0.5ex] 
\hline\hline
\texttt{cause-effect} & 10.7K & 12.2 & Relations \\ 
\hline
\texttt{musicians} & 15.8K & 10& Entities \\
\hline 
\texttt{directions} & 15.3K & 3.8& Intents \\ 
\hline
\texttt{profession} & 1M & 1.1& Entities  \\
\hline
\texttt{tweets} & 2130 & 11.4 (Food) & Intents \\
\hline
\end{tabular}
\end{center}
\vspace{-2em}
\caption{Dataset statistics \label{tab:datasets}}
\end{table}

\noindent\textbf{Datasets.} We experimented with five diverse real-world datasets
each suitable for one of the following NLP tasks: entity extraction, relationship
extraction, and intent classification. Table~\ref{tab:datasets} summarizes the
statistics of these datasets. All datasets, except for \texttt{directions}, come
with ground-truth labels which we use for evaluation and to synthesize the responses
from an oracle. For the \texttt{directions} datasets, we rely on human annotators to generate the gold standard and validate
the heuristics. 
We describe each of these dataset below.
\begin{compactitem}
    \item \texttt{cause-effect}~\cite{semeval} is a dataset commonly used as a
    benchmark for relationship extraction between pairs of entities. We focus on the
    task of finding sentences that describe a cause and effect relationship between
    two entities.
    \item \texttt{directions} is an internal dataset described in
    Example~\ref{ex:ex1}. For this dataset, we leveraged
    Figure-eight\footnote{https://figure-eight.com} crowd workers to verify the
    heuristics generated by \darwin.
    \item \texttt{musicians} dataset consists of sentences from Wikipedia articles. The task is entity extraction with the goal to extract
    the names of musicians. The ground-truth is obtained with the help of
    NELL's knowledge-base\footnote{\url{http://rtw.ml.cmu.edu/rtw/kbbrowser/}}. 
    \item \texttt{professions} dataset is a collection of sentences
    from ClueWeb\footnote{\url{https://lemurproject.org/clueweb09/}}. The sentences  that
    mention various professions (e.g., scientist, teacher, etc.) are positives. The ground truth is generated using NELL's knowledge-base. 
    \item \texttt{tweets}~\cite{aaai15tweets} data set is a benchmark for classifying
    the intent of tweets into predefined categories such as food, travel and career, etc. 
\end{compactitem}

\smallskip
\noindent \textbf{Baselines.} We evaluate our framework on two fronts: (1) the ratio of
positive instances it discovers (i.e. coverage) and (2) the performance of the classifier
trained using our weakly-supervised labels. Our baselines for these two evaluation criteria
are listed below. 
\begin{compactitem}
    \item Section~\ref{sec:Snuba} compares the fraction of positives identified with \texttt{Snuba}\cite{Snuba}. In this experiment, we consider a small sample of positives chosen randomly from the dataset.  
    \item Section~\ref{sec:rulecoverage} compares the coverage obtained by \darwin\ against
    two baselines, namely  \texttt{HighP} and \texttt{HighC}.  \texttt{HighP} is a simpler
    version of \darwin\ which selects the rule which is expected to have a high
    precision (according to the classifier) and submits it to the oracle.
    On the other hand, \texttt{HighC} selects rules with the maximum coverage,
    irrespective of their expected precision\footnote{\texttt{HighC}'s performance
    was quite poor as most of its suggested rules are rejected by the oracle. As
    a result, we omit HighC from the plots for the sake of clarity.}.
    \item Section \ref{sec:classifierquality} compares the F-score of the classifier
    generated by \darwin\ with an \emph{Active Learning} (\texttt{AL}) \cite{activelearning}
    and a \emph{Keyword Sampling} (\texttt{KS}) technique as well as the \texttt{HighP}
    baseline mentioned earlier. \texttt{AL} improves its performance by selecting
    the instance with the highest entropy and asking the oracle for its label.
    It then re-trains the classifier using the new label.
    The \texttt{KS} approach is designed to check if we can quickly obtain a small
    set of promising instances by filtering the corpus using a set of relevant keywords,
    and label the instances in the smaller set. To do so, we asked annotators to provide 10
    distinct keywords as a heuristic to filter the dataset. The \texttt{KS} technique
    randomly samples instances from the filtered dataset and asks for its label.
    We employ the same deep learning based classifier for all the techniques. 
\end{compactitem}
Finally, note that \darwin\ can use different traversal algorithms: \babystep,
\universal, and \hybrid, which we refer to as 
\darwinp{LS}, \darwinp{US}, and \darwinp{HS} respectively.

\smallskip
\noindent\textbf{Settings.} 
We implemented all proposed algorithms and baselines in Python and ran
the experiments on a server with a 500GB RAM and 64 core 2.10GHz
x 2 processors. The dependency parse trees and the POS tags are
generated with SpaCy\footnote{\url{https://spacy.io/}}. All text
classifiers trained in our experiments (whether used by \darwin\ or
other baselines) are implemented with a 3-layer convolutional neural
network followed by two fully connected layers, following the
architecture described by Kim et al~\cite{kim14convolutional}.
The input to the classifier is a matrix created by stacking the
word-embedding vectors of the words appearing in the sentence. We also
used SpaCy's word-embeddings for
English\footnote{\url{https://spacy.io/models/en\#en_core_web_lg}}.
For generating
derivation sketches, the maximum depth is set to 10 and we consider 10K heuristics in candidate selection.
When simulating the responses from an oracle (using the
ground-truth data), we respond YES to heuristic $h$ if at least
80\% of its coverage set consist of positive instances.

\begin{figure}
\centering
\vspace{-.2em}
\subfloat[\texttt{directions}\label{fig:direc_Snuba}]{\includegraphics[width=0.24\textwidth]{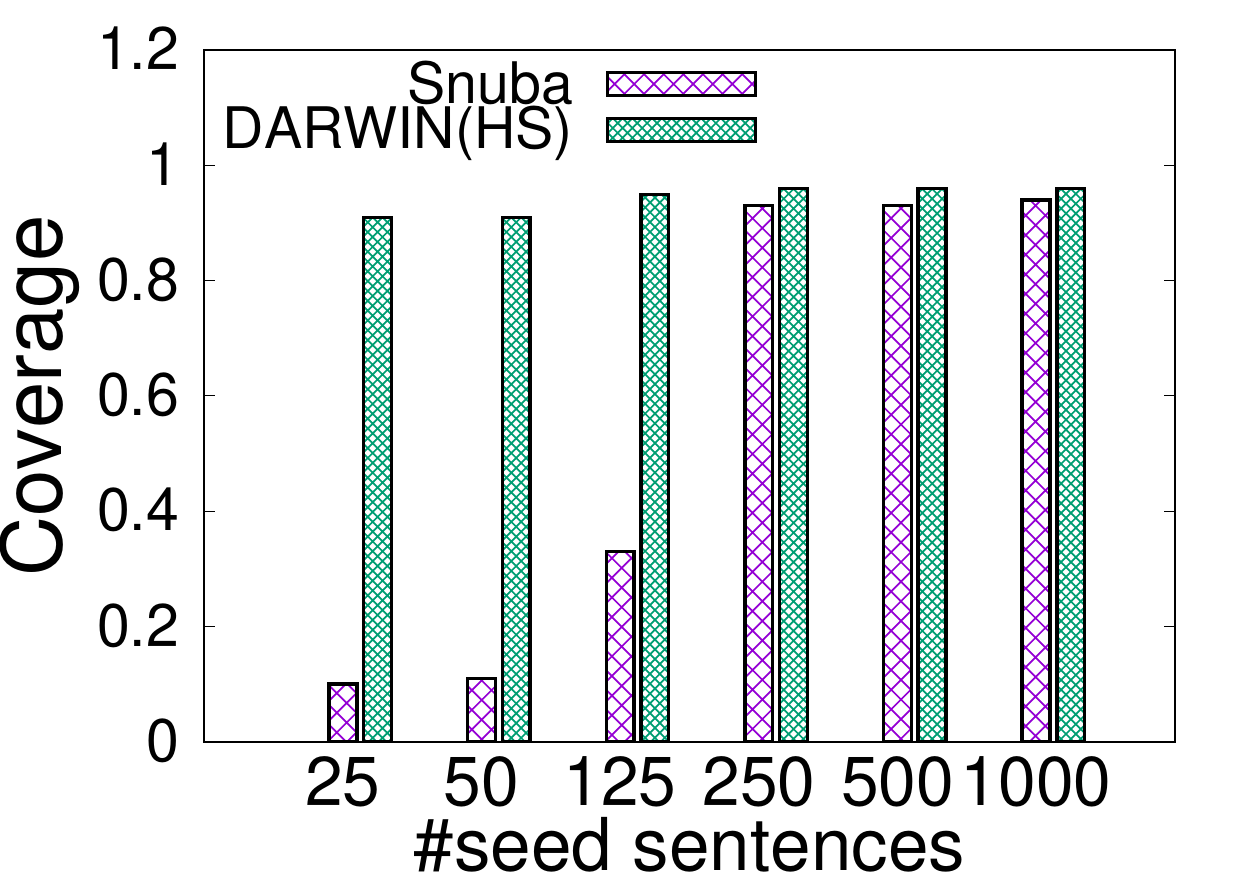}} ~
\subfloat[\texttt{musicians}\label{fig:music_Snuba}]{\includegraphics[width=0.24\textwidth]{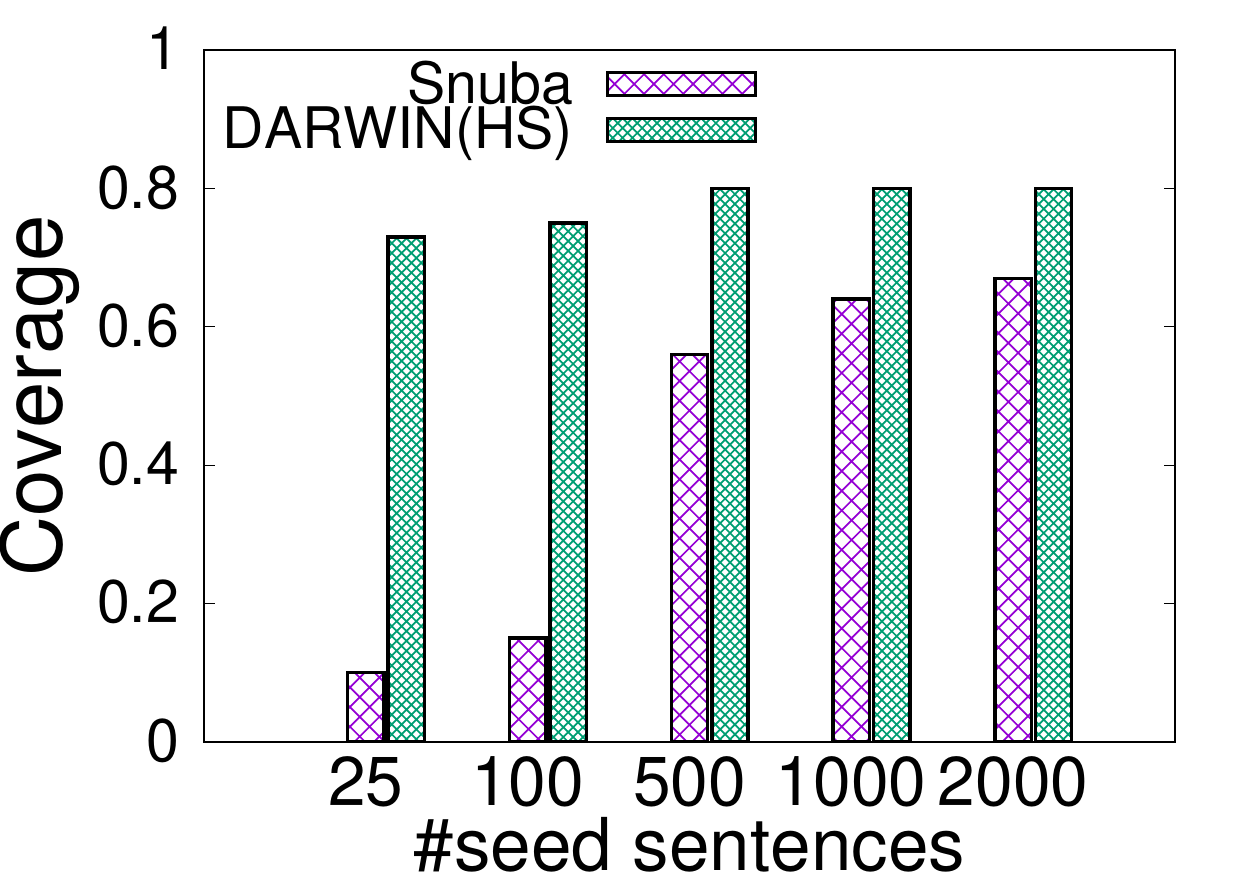}}
\caption{Effects of seed set size on performance.\label{fig:Snuba_random}}
\vspace{-1.2em}
\end{figure}

\begin{figure}
\centering
\subfloat[\texttt{directions}\label{fig:direc_Snuba}]{\includegraphics[width=0.24\textwidth]{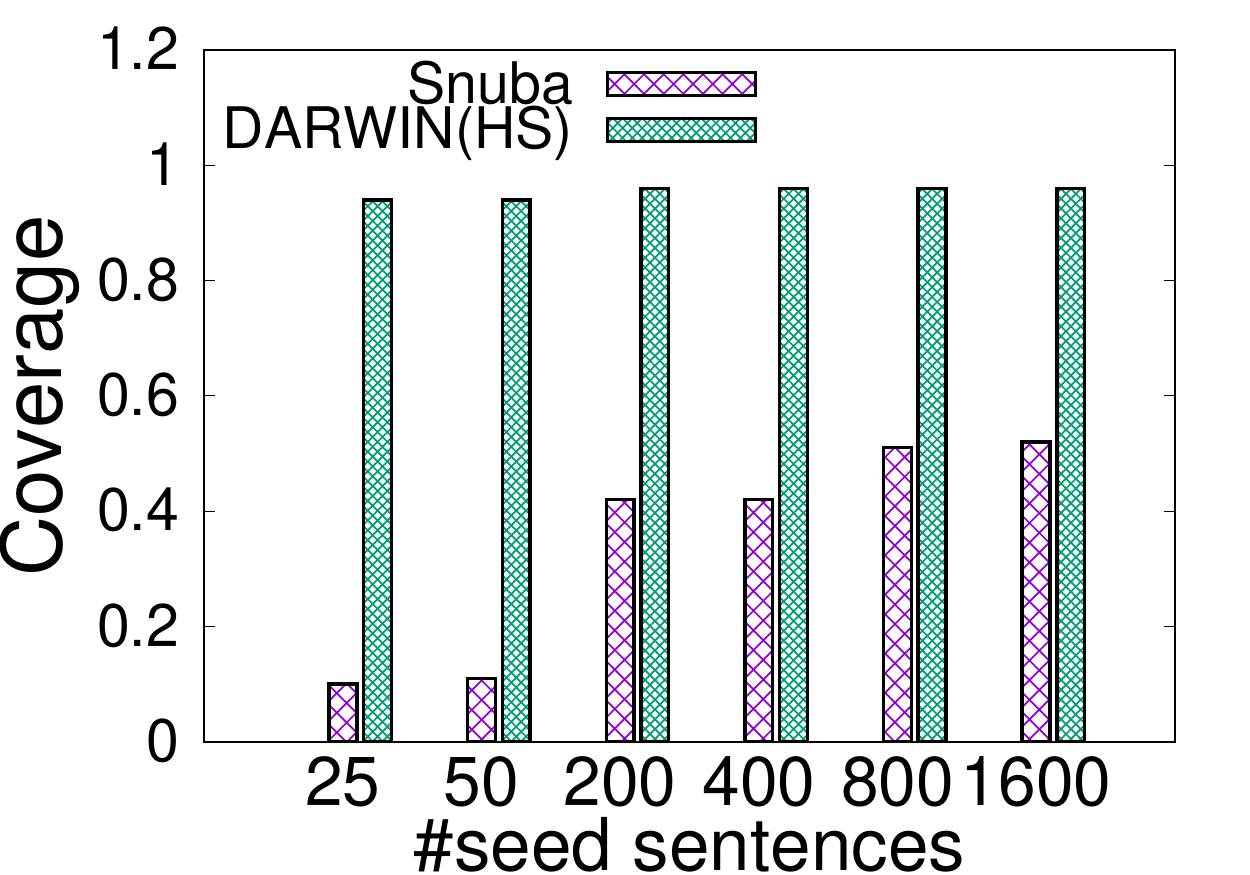}} ~
\subfloat[\texttt{musicians}\label{fig:music_Snuba}]{\includegraphics[width=0.24\textwidth]{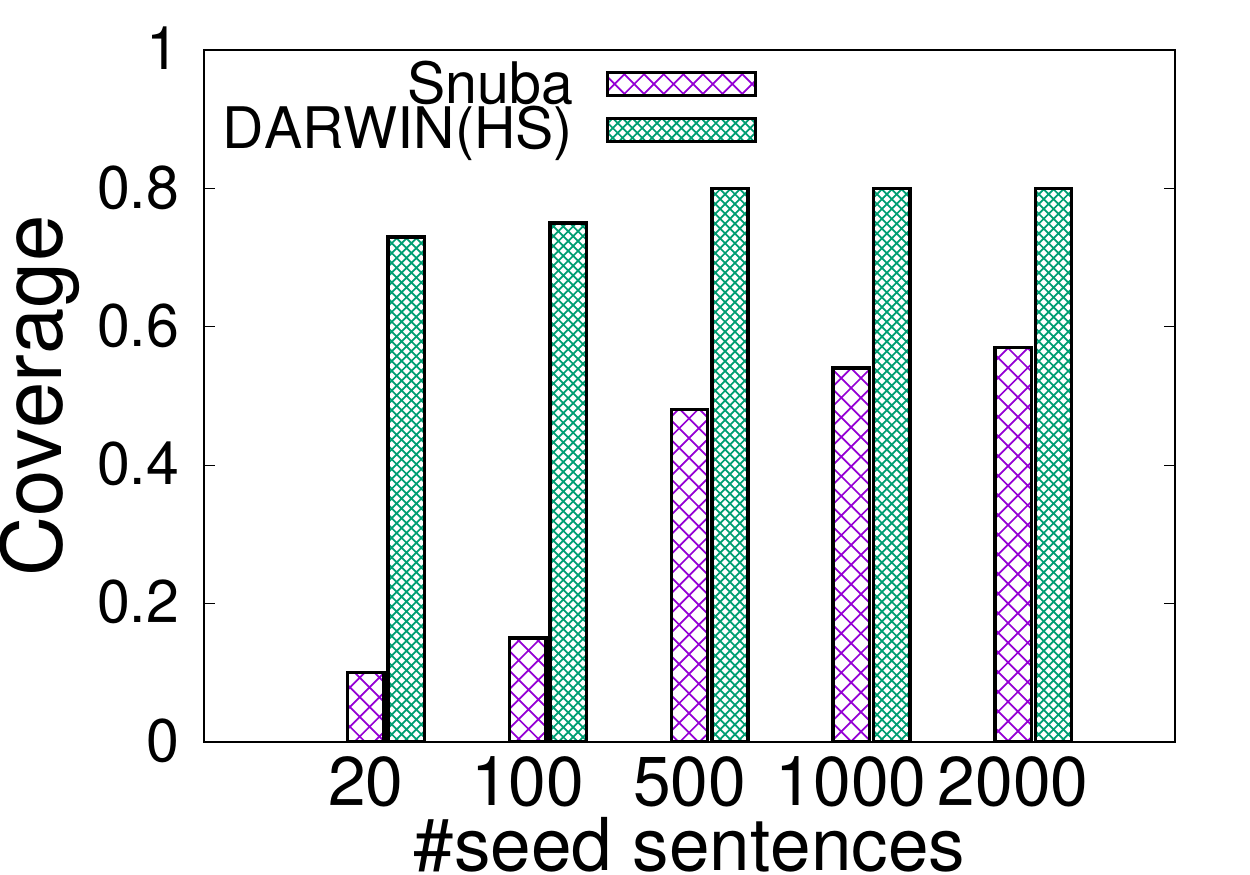}}
\caption{Effects of biased seed set size on performance.\label{fig:Snuba_biased}}
\vspace{-0.9em}
\end{figure}

\begin{figure*}[t]
\centering
\includegraphics[width=\textwidth]{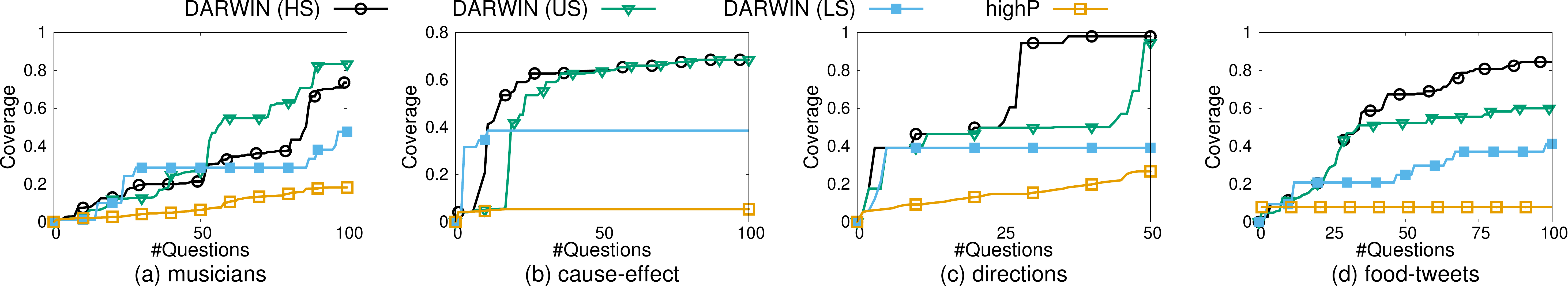}
\includegraphics[width=\textwidth]{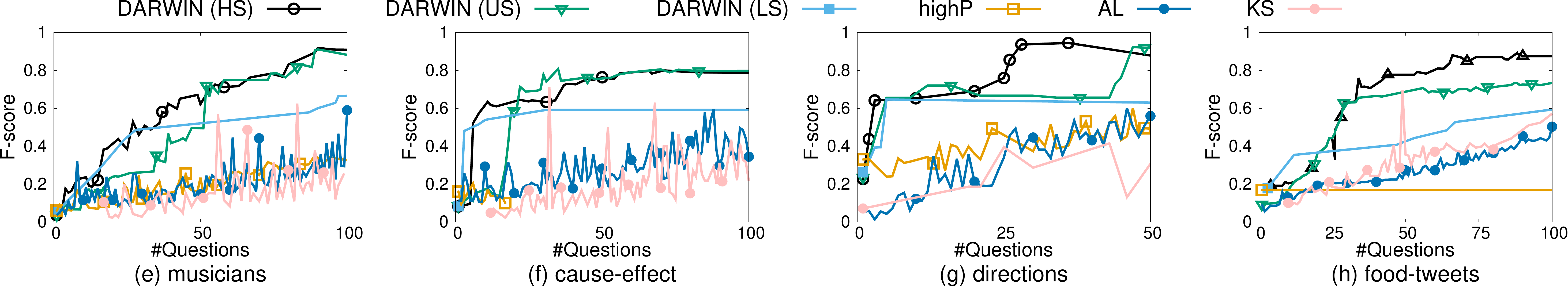}
\caption{Comparison of rule coverage and classifier's F-score for \darwin{} based pipelines on various datasets. \label{fig:overall}}
\vspace{-0.9em}
\end{figure*}
\begin{figure}[t]
\centering
\includegraphics[width=0.48\textwidth]{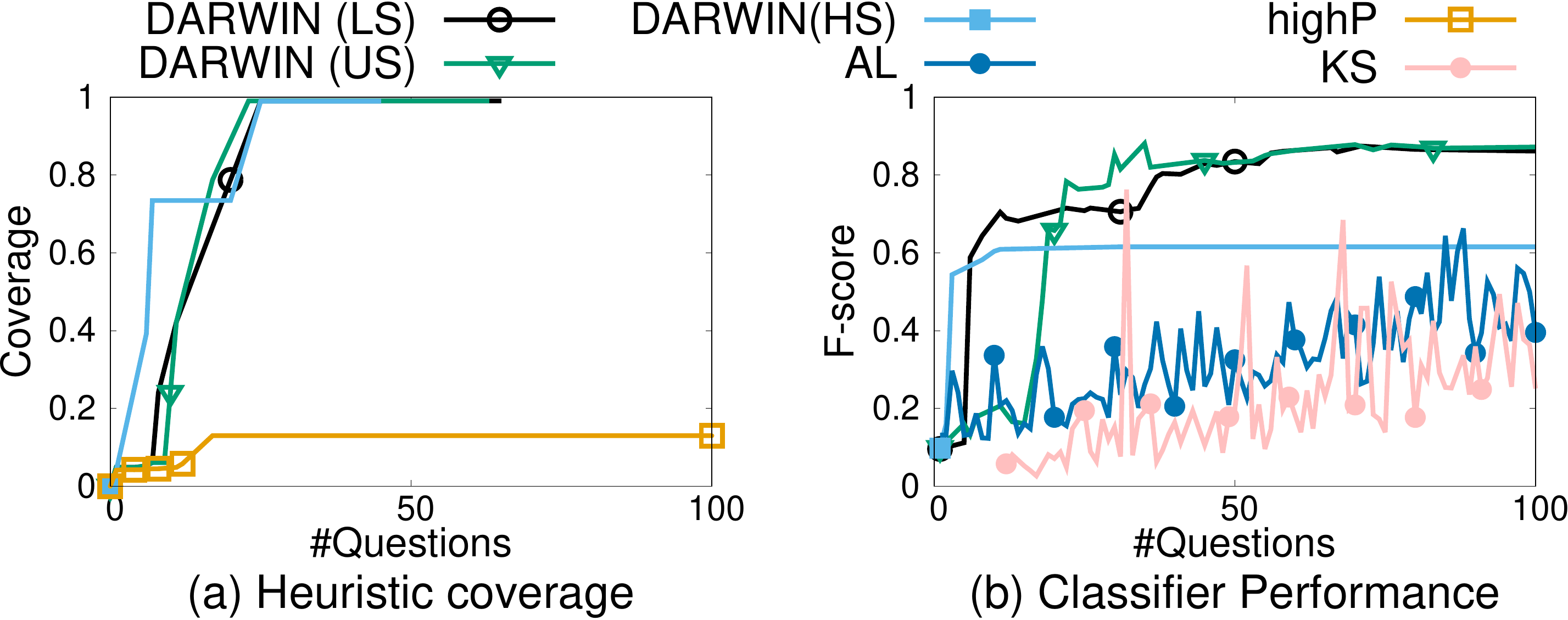}
\caption{Comparison for \texttt{profession}.\label{fig:profession}}
\vspace{-0.5em}
\end{figure}

\subsection{Comparison with Snuba\label{sec:Snuba}}

In this experiment, we initialize \texttt{Snuba} and \darwinp{HS} with the same set of randomly chosen labeled sentences and compare the total number of positives identified by each of the techniques\footnote{In this experiment, we do not start with a single labeling heuristic as \texttt{Snuba} achieves a very small coverage as it fails to obtain enough positive instance
due to the high degree of imbalance in these datasets.}. Note that \texttt{Snuba} does not
query the oracle and infers heuristics based on the provided labeled instances may infer
inaccurate heuristics. For a fair evaluation, we choose not compare the accuracy of identified heuristics. Figure~\ref{fig:Snuba_random} shows the change in fraction of identified positives by
varying the size of the initial seed set. \darwinp{HS} is able to identify majority of positives even when the pipeline is initialized with less than $25$ sentences. However, \texttt{Snuba} requires at least $200$ randomly chosen sentences for \texttt{directions} and $1000$ for \texttt{musicians}. If we employ expert to sample positives, \texttt{Snuba} requires at least 100 positive samples in \texttt{musicians}.

To further evaluate the generalizability of \texttt{Snuba} and \darwinp{HS} to identify heuristics that have limited or no evidence in the initial seed set, we construct a biased sample of seed positives. In this experiment, we choose sentences randomly from the corpus after ignoring the ones that contain the token `shuttle' in \texttt{directions} dataset and `composer' in \texttt{musicians}. Figure~\ref{fig:Snuba_biased} shows the fraction of positives identified with varying size of the seed set. \texttt{Snuba} is not able to identify the positives that contain the token `shuttle' in \texttt{directions} and `composer' in \texttt{musicians}. Henceforth, it achieves poor coverage over the positives in two datasets. \darwin{HS} is able to identify majority of positives irrespective of the number of sentences used to initialize the pipeline. \texttt{Snuba} requires considerably more labelled sentences in \texttt{musicians} as compared to \texttt{directions} due to the presence of many diverse heuristics in the dataset, 
most of which have limited evidence in the seed subset. We observe similar performance gap between \texttt{Snuba} and \darwinp{HS} for other datasets. 

This experiment validates that \texttt{Snuba} works well when the initial  seed set has enough randomly chosen positives and lacks the ability to generalize to heuristics that have limited evidence. On the other hand, \darwinp{HS} is able to identify majority of the positives even when the pipeline is initialized with just $25$ sentences and has good generalizability. To further evaluate \darwin{}'s ability to identify positives, the following subsection considers a more challenging scenario where the pipeline is initialized with a single labeling heuristic or just two positive sentences.

\subsection{Rule Coverage\label{sec:rulecoverage}}
Figures 6a-6d and 10a illustrate the fraction of positives identified by \darwin\ and our baselines. We can observe that \darwinp{HS} has the most stable
performance and outperforms other techniques. While \darwinp{US} occasionally
outperforms \darwinp{HS}, we observe that it fails to perform well on all datasets.
In most cases (with an exception of \texttt{cause-effect}), the \darwinp{HS}
achieves a coverage of 0.8 using less than 120 queries to the oracle.
The \texttt{cause-effect} is known to be a tough benchmark in the NLP community
as the best F-score reported by \cite{semeval}
is 82\% given complete access to the training set. Assuming
that the oracle considers a majority vote by querying three crowd members and
each query costs 2 cents\footnote{These are standard assumptions in crowdsourcing
platforms eg. figure-eight. We used the same cost model to collect labels for
\texttt{directions}.}, the \darwinp{HS} pipeline generates more than 80\% of the
positive labels with only $\$7.20$.  Figure 6d demonstrates the
behavior for `Food' intent in the tweets. We observed similar behavior  for `Travel'
and `Career' intents on this data set. We can observe that the other baselines do
not perform well compared to \darwin; The \texttt{highP} identifies heuristics with
very small coverage as its candidates. Also note that the \darwinp{LS} algorithm
shows a high progressive coverage initially but it converges to a very low coverage
value because it is unable to identify rules that are semantically similar, but
 far away in the hierarchy. { Overall, we recommend \darwinp{HS} for any practical application as it is more robust and works better than most of the techniques. On the other hand,  \darwinp{LS} and \darwinp{US} variants work well in specific settings.
 \darwinp{LS} performs better than the other techniques when precise rules are present close to each other in the hierarchy and  \darwinp{US} performs well in the presence of abundant labelled examples.}

Figure \ref{fig:traversalex} shows some the heuristics which are queried by the
\darwinp{HS} algorithm. In the \texttt{directions} example, \darwinp{HS} started with
`\emph{best way to get to}' and was able to traverse to `\emph{shuttle to}', which
is quite distinct from the initial seed rule. The choice of `\emph{to the hotel
from}' by the algorithm provides some evidence that `\emph{shuttle to}' is also
a good rule since the phrases often co-occur together in positive instances. In
the \texttt{cause-effect} example, the traversal is relatively simple as the
algorithm generalizes the initial rule first and as soon as it reaches the noisy
and unhelpful rule `\emph{by}', it again specializes to `\emph{triggered by}' which
is again a precise rule.  In addition to these simpler heuristics, \darwin{} identified more complex heuristics for \texttt{professions} like `/is/NOUN $\wedge$ job', among others.

\begin{figure}[t]
\centering
\includegraphics[scale=0.32]{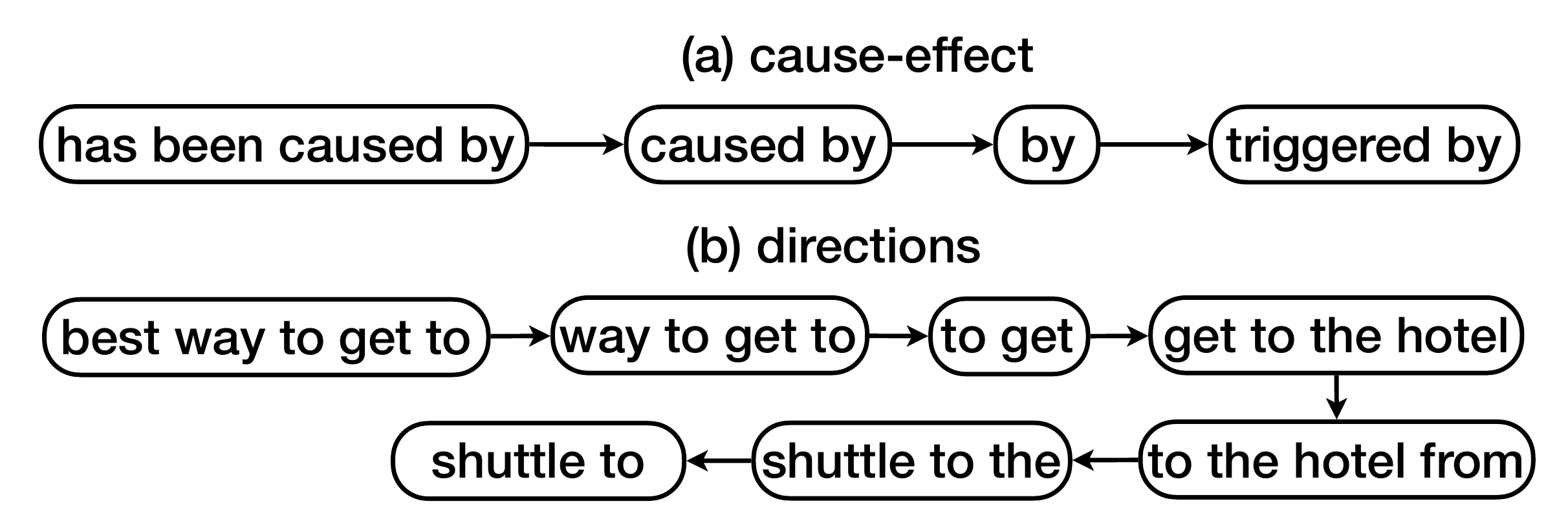}
\caption{Example of traversals by \hybrid{} algorithm on
two datasets.\label{fig:traversalex}}
\vspace{-0.5em}
\end{figure}


\vspace{-2mm}
\subsection{Quality of the Classifier\label{sec:classifierquality}}
This section compares the quality of the classifier generated using the labels 
identified by \darwin{}. Figure 6e-6h and 10b show that \darwinp{HS} dominates other techniques over all the datasets. The active learning technique suffers
from poor F-score initially and improves gradually. Since AL generates very few
training examples, the trained classifier is highly unstable and shows jittery
F-score.  The KS approach shows similar performance and performs comparable to AL.
On the other hand,
\darwin{} based pipelines are much more stable in terms of F-score. The classifier
that was trained with the labeled data generated by \darwin{} pipelines always
maintains  a high precision. It is interesting to note that \cite{aaai15tweets}
reports the maximum F-score for food intent to be 0.54, as compared to 0.84 by
\darwin{}. The classifier generated by \darwin{} achieved an F-score above 
0.8 for other intents like `Travel' and `Career' too while \cite{aaai15tweets} reports a maximum F-score of 0.58 for these intents.

\vspace{-2.5mm}
\subsection{Additional Experiments}
To provide better insights into \darwin's performance, we have conducted a series of experiments
to evaluate (1) how efficient the framework is in terms of the time required to obtain labels,
(2) how much noise-aware models (trained by Snorkel) can improve the classification results,
(3) how well human annotators approximate our notation oracles.
Due to space limitations, we present the effect of varying seed rules and parameters in the Appendix.

\smallskip
\noindent \textbf{Efficiency in Label Collection.}
As we demonstrated, \darwin{} identifies majority of the positive 
instances
in all the datasets using roughly 100 queries. The time taken to generate the
index structure for all the datasets was less than 5 minutes. The hierarchy
generation phase then iterates over the index to identify the candidate
rules. This phase takes less  15 minutes for a corpus of 100K. 

Since the \babystep{} algorithm does not require the index to be
pre-computed and generates candidates on the fly, it runs in less
than 45 minutes for all datasets. \hybrid{} and \universal{}
traversal algorithms require 60-90 minutes on smaller datasets
(i.e., \texttt{directions}, \texttt{musicians} and
\texttt{cause-effect}) and about 2 hour and 45 minutes on
\texttt{professions}. The major bottleneck in this process is the time
taken by the classifier to make a prediction for 
all instances in the corpus (It takes roughly 25 minutes for one round
of training and testing on the \texttt{professions} dataset).
We implemented a simple optimization where we
evaluated a sentences only if it had a confidence score more than 0.3 in
the previous iteration and only evaluated instances that didn't
satisfy this constraint once every three iterations. This heuristic
helped us reduce the running time from 2 hours and 45 minutes to 65
minutes for the \texttt{professions} dataset. The total running time does not grow linearly with the size of dataset because most of the components use the classifier to identify the positives. These candidate positives are used for hierarchy generation and traversal. Hence, the running time grows linearly with positive set size but not with dataset size (after using the above mentioned optimization).

\smallskip
\vspace{-1mm}
\noindent \textbf{Training noise-aware classifiers.}
One of the recent developments in weak-supervision paradigms has been emergence of 
frameworks such as \emph{Snorkel}~\cite{snorkel} which are designed to de-noise the
generated labels and train noise-aware classifiers. In this experiment, we direct the
set of rules identified by \darwin\ to Snorkel, and compare the quality of the noise-aware
classifier against a classifier trained directly on the labels generated by \darwin.
Table \ref{tab:snorkel} summarizes  the F-score that two classifiers have obtained on 
our datasets. We can observe that in most cases, using Snorkel does not yield any improvements.
This is mainly because in many of these datasets, the rules generated by \darwin\ already exhibit
a low degree of noise and a good coverage and thus there is almost no room for improvements.
Nevertheless, we can see that on some datasets such as \texttt{directions} using Snorkel can be
quite beneficial.

\begin{table}[t]
\small
\begin{center}
\begin{tabular}{c|cccc} 
\multicolumn{1}{c}{} &\texttt{M}& \texttt{C} &
\texttt{D} & \texttt{F} \\ [0.5ex] 
\hline\hline
\darwin{} & 0.91 & 0.79 & 0.89 & 0.87 \\ 
\hline
\darwin{}+Snorkel & 0.82 & 0.78 & 0.97 & 0.87 \\
\hline 
\end{tabular}
\end{center}
\vspace{-0.8em}
\caption{Performance of \darwin{} with Snorkel
(\texttt{M=musicians,C=cause-effect,D=directions,F=food-tweet})\label{tab:snorkel}}
\end{table}

\smallskip
\noindent \textbf{Performance of human annotators.}
Clearly, \darwin's performance heavily relies on the quality of responses it receives
from the annotators. To study how well human annotators perform, we ran an experimental
study on Figure-eight crowd-sourcing platform for \texttt{directions} dataset. Labels were
collected for 2\,600 heuristics. 
Each annotator was paid 2 cents per a single rule evaluation and three evaluations per
rule were collected.
A manual inspection of the results reveals that annotators were able to capture most of
the precise heuristics such as `best way to get there', `shuttle from', `across the street from',
`airport to hotel', and etc. Overall, we found less than 10 false positives responses in
the 69 positive heuristics identified by the crowd labels. These erroneous responses were due
to the fact that the 5 matching sentences presented to the annotator sometimes can have 3 or
4 positive instances by chance which confuses the annotators. Presenting more
samples lowers the error rate. Interestingly, \darwin\ often rates these heuristics lower in
preference to query as it can analyze the complete coverage set, and mitigate such errors
by considering the entire distribution of instances. The annotators took 23 sec on average to label a heuristic query. For 100 queries, \darwin\ generates all the labels in less than 40 min
of human effort. This time can further be reduced by asking various questions in
parallel to different crowd members.

\vspace{-2mm}
\section{Related work}
\label{sec:related}
To the best of our knowledge, \darwin\ is the first system that
assists annotators to discover rules under any desired rule
grammar for rapid labeling of text data. Our work is related to
studies in areas of weak supervision, crowdsourcing, and the
intersection of the two which we discuss next.

\smallskip
\noindent \textbf{Weak Supervision.}
There are multiple existing approaches for generating labels in
weakly supervised settings. Some techniques rely on the notion
of \emph{distant supervision} where the labels are inferred using
an external knowledge
base~\cite{mintz09distant,alfonseca12pattern,kb1}. One notable
example is a system named Snuba \cite{Snuba} which generates
labeling rules based on an existing labeled dataset. In contrast
to these systems, \darwin\ is designed for scenarios where no
additional sources of information are available. In such cases,
it is necessary to rely on annotators to write labeling rules.
While using expert-written rules have proven to be highly effective
in many settings~\cite{snorkel}, there is limited work on how to
facilitate the process of writing or discovering high-quality rules.
One interesting example is Babble Labble\cite{snorkel2}, a
labeling tool that allows annotators to explain (in natural
language form) why they have assigned a label to a given
data point. These explanations are then transformed into
labeling rules. While Babble Labble simplifies the rule writing
process, it only handles a single internal rule language. On the
other hand, \darwin\ allows experts to pick their desired rule
language depending on the complexity and the dynamics of the
task at hand. 

There have been several studies on utilizing the weakly-supervised
labels in an optimal way. Snorkel~\cite{snorkel} and
Coral~\cite{varma17inferring} are recent examples of systems
(based on the data programming paradigm) that de-noise and utilize
the labels collected via weak supervision.
Similarly, there are numerous data management problems spanning
data fusion \cite{fusion1,fusion2} and truth discovery
\cite{truth}, which focus on identifying reliable sources
of data. Many recent studies in data integration have also
explored techniques that handle error in crowd answers
\cite{crowder1,crowder2}. Note that \darwin\ is a framework
for discovering labeling rules which goes hand-in-hand with
the aforementioned systems since \darwin's generated rules can
be further processed using these de-noising techniques to
achieve better results.

\smallskip
\noindent \textbf{Crowdsourcing Frameworks.}
There has been many studies on devising oracle based abstractions
that handle annotations from a crowd and minimize the noise
in answers \cite{crowd_error,crowd_error2}. Perhaps, more relevant
to our work, are existing studies on how labeling rules can
be verified with the help of the crowd. One recent example is
a system named CrowdGame \cite{crowdgame} which validates
a rule by showing either the rule or its matching instances to
the annotators. The authors demonstrate that their
proposed game-based techniques yields the best
results for rule verification. Unlike \darwin\, CrowdGrame
assumes a pre-existing (manageable) set of possible rules from
which the best rule should be selected. \darwin, on the other
hand, has no such assumption and has to create a promising
set of rules from the rule grammar. Additionally, the
game-based approach to annotate a rule can be modeled as an Oracle
in \darwin.

\vspace{-2mm}
\section{conclusion}
\label{sec:conclusion}
We present \darwin{}, an interactive end-to-end system that enables annotators to
rapidly label text datasets by identifying precise labeling rules for the task at hand.
\darwin{} compiles the semantic and syntactic patterns in the corpus to generate
a set of candidate heuristics that are highly likely to capture the positives instances
in the corpus. The set of candidate heuristics are organized into a hierarchy which
enables \darwin\ to quickly determine which heuristic should be presented to the annotators 
for verification. Our experiments demonstrate the superior performance of \darwin{}
in wide range of labeling tasks spanning intent classification, entity extraction
and relationship extraction.

\newpage
\bibliographystyle{abbrv}
\bibliography{darwinref}

\clearpage
\appendix
\section{Proof of Lemma 4}
\begin{proof}
Expected score of the heuristic function is 
\begin{eqnarray*}
E[\sum_{s\in C_r} X_s] &=& \sum_{s\in C_r} \mu_s\\
&=& \sum_{s\in C_r\cap P} \mu_s +\sum_{s\in C_r\setminus P} \mu_s  \\
&\leq & \sum_{s\in C_r\cap P} (\beta+(1-\theta)(1-\beta)) \\
&&+\sum_{s\in C_r\setminus P} \left( \beta'+(1-\theta)(1-\beta')\right)   \\
&=&  \left( \beta+(1-\theta)(1-\beta)\right)p|C_r| \\
&& + (1-p)|C_r|\left( \beta'+(1-\theta)(1-\beta')\right)   \\
&\le&\left( \beta+(1-\theta)(1-\beta)\right)|C_r|
\end{eqnarray*}
\end{proof}

\section{Proof of Lemma 5}
\begin{proof}
The score of heuristic function $r$ is $\sum_{s\in C_r} X_s$. The expected value of the score is calculated in lemma \ref{lem:exp2}. Using Hoeffding's inequality,
\begin{eqnarray*}
Pr[\frac{1}{|C_r|}\sum_{s\in C_r} X_s \leq  (1+\epsilon)\mu_r/|C_r| ] &\leq& 2e^{-2\epsilon^2\mu_r^2/|C_r|}\\
&=& 2e^{-2\epsilon^2 (\beta+(1-\theta)(1-\beta))^2 |C_r|}\\
&=& 2e^{-4\log n}=\frac{2}{n^4}\\
\end{eqnarray*}
This shows that $\frac{1}{|C_r|}\sum_{s\in C_r} X_s$ is smaller than $(1+\epsilon) (\beta+(1-\theta)(1-\beta))|C_r|$ with a probability more than $1-\frac{2}{n^4}$
\end{proof}

\section{Proof of Lemma 6}

\begin{proof}
Using Lemma~\ref{lem:s1} and \ref{lem:s2}, we can observe that the calculated benefit of heuristic $h_1$ is atleast $(1-\epsilon) \theta \beta' |C_{r_1}| $ with a probability of $1-\frac{2}{n^4}$. Similarly, the score of $r_2$ is atmost $ (1+\epsilon) \left( \beta+(1-\theta)(1-\beta)\right)|C_{r_2}|$ with a probability of $1-\frac{2}{n^4}$.
This shows that 
\begin{eqnarray}
score(r_1)&>&score(r_2)\\
(1-\epsilon) \theta \beta' |C_{r_1}|  &>& (1+\epsilon) \left( \beta+(1-\theta)(1-\beta)\right)|C_{r_2}|\\
\frac{|C_{r_1}|}{|C_{r_2}|} &>& \frac{(1+\epsilon)\left( \beta+(1-\theta)(1-\beta)\right)}{(1-\epsilon)\theta \beta' } \\
\frac{|C_{r_1}|}{|C_{r_2}|} &>& \alpha
\end{eqnarray}
where $\alpha$ is a constant.
\end{proof}
\section{Additional Experiments\label{sec:appendixa}}

\begin{figure}
\centering
\subfloat[\texttt{musicians}\label{fig:tau}]{\includegraphics[width=0.24\textwidth]{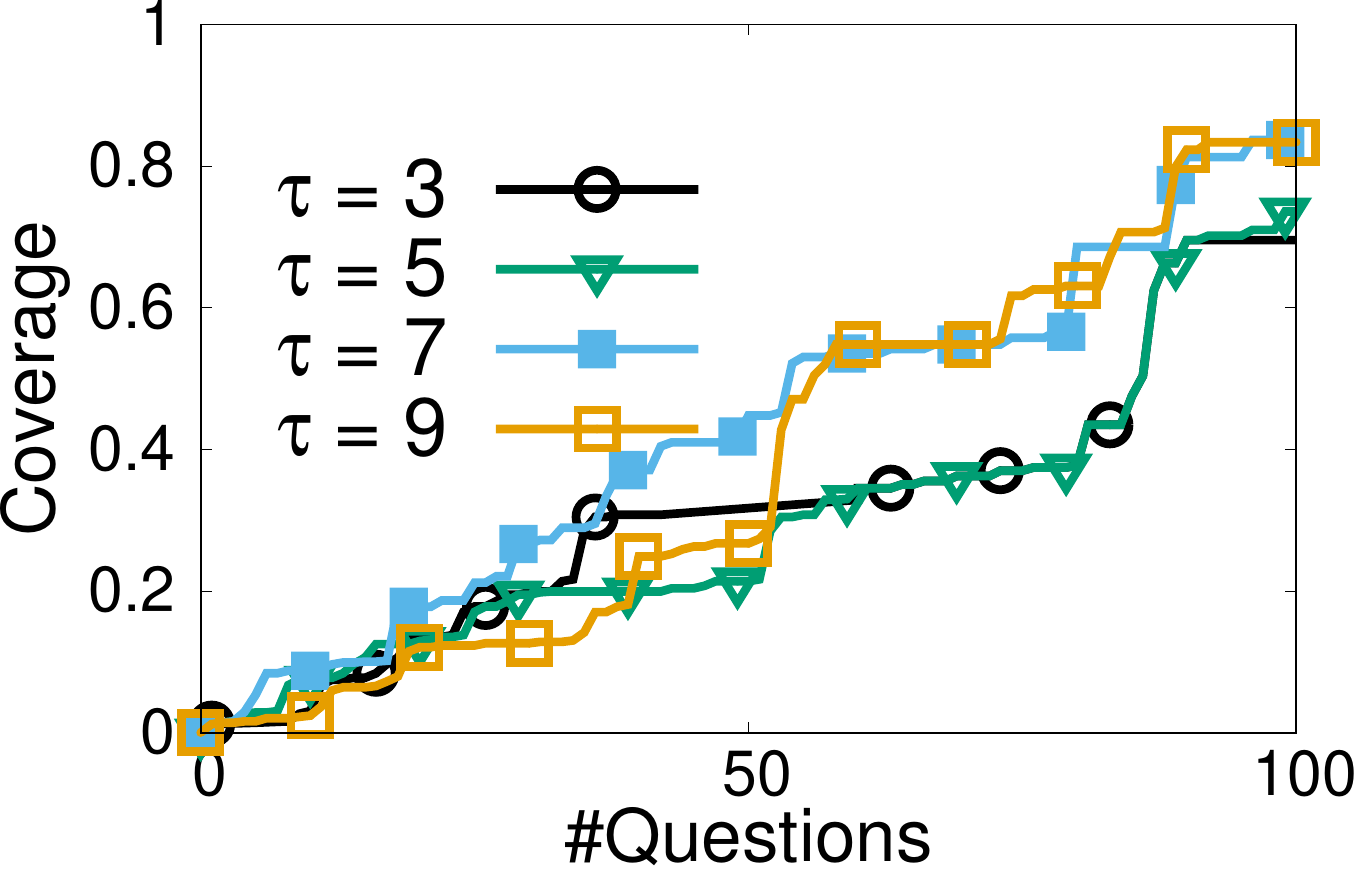}} ~
\subfloat[\texttt{musicians}\label{fig:rule}]{\includegraphics[width=0.24\textwidth]{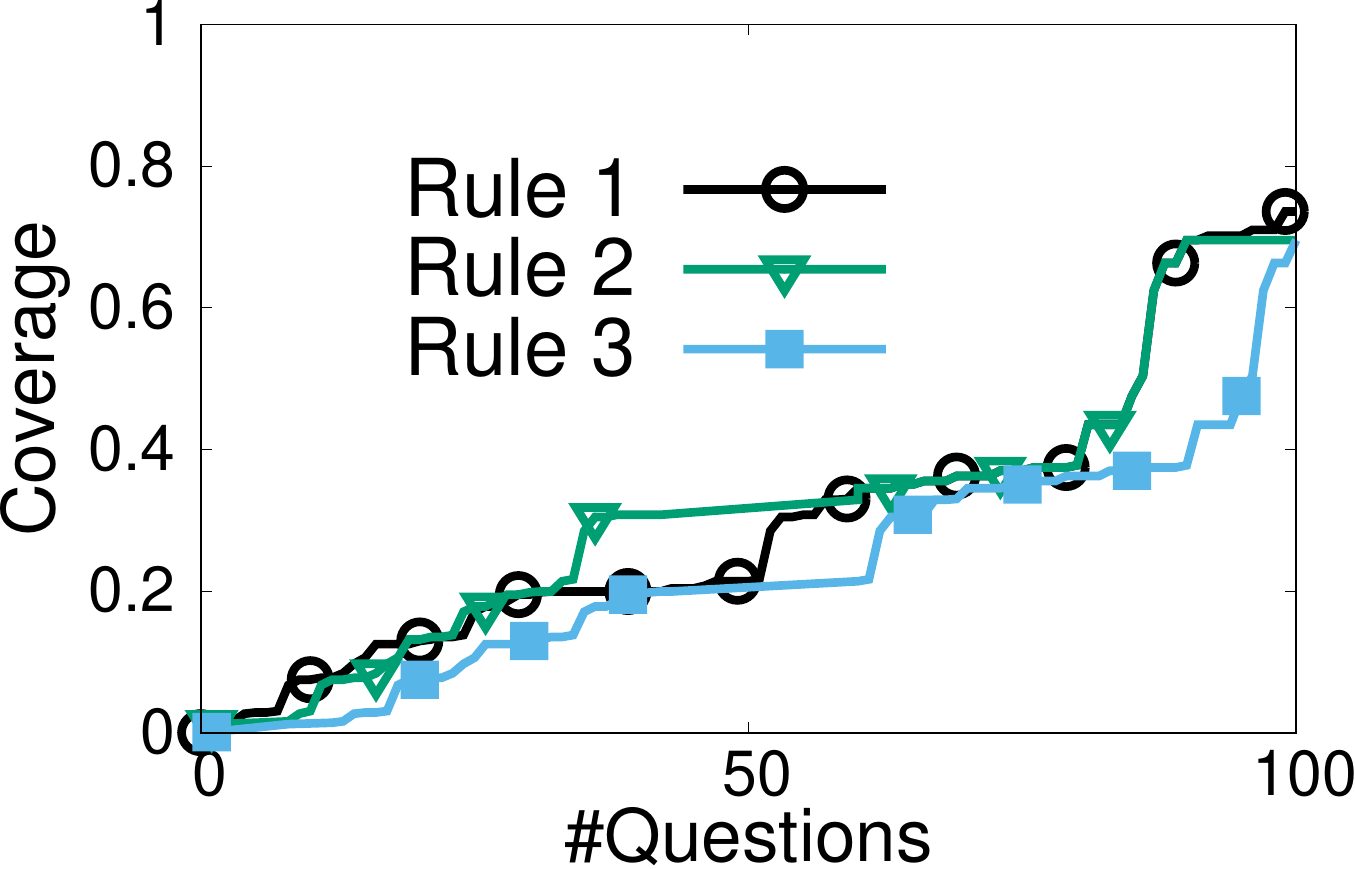}}
\caption{Sensitivity of \darwin{} to $\tau$ and seed rules.}
\end{figure}

\smallskip
\noindent \textbf{Sensitivity to \hybrid's traversal parameters.}
Here, we study to what extent \darwin's performance is sensitive
    to parameter $\tau$ in the \hybrid traversal algorithm. Recall that
parameter $\tau$ determines how often
the \hybrid\ algorithm switches between exploiting the local structure of the
hierarchy as opposed to evaluating all possible candidates using
the classifier.
Figure~\ref{fig:tau} shows that \darwinp{HS} performs very similar on varying
$\tau$. The solution quality tends to improve slightly on increasing $\tau$
because the effective rules for the \texttt{musicians} dataset are not
close to each other in the hierarchy. However, note that choosing
large values of $\tau$
can affect the efficiency of the pipeline. More precisely, large
values of $\tau$
force the \hybrid\ system to rely on its internal classifier
to evaluate all rule candidates for too many steps which can be
quite time consuming. 

\smallskip
\noindent \textbf{Sensitivity to seed rule.} 
This experiment establishes that \darwin{} has a robust performance given different
types of input seed rule. Focusing on the \texttt{musicians} dataset, we initialize \darwin\ with the following seed rules. Rule
1 is the keyword `\emph{composer}' stating that any sentence containing this word
mentions a musician. Rule 2 is the keyword `\emph{piano}' and finally Rule 3 is
the sentence `\emph{Beethoven taught piano to the daughters of Hungarian Countess
Anna Brunsvik.}'. Note that Rule 2 is an extremely generalized version of Rule
3. Figure \ref{fig:rule} compares the performance of \darwinp{HS} for all three
seed rules. \darwin{} performs equally well on three different types of input
seed rules. We can observe that for Rule 3, \darwin{} requires the initial 8
queries to generalize the seed rule, and as soon as it identifies a rule with
high coverage, it performs very similar to the other seed rules.

\begin{figure}
\centering
{\includegraphics[width=0.24\textwidth]{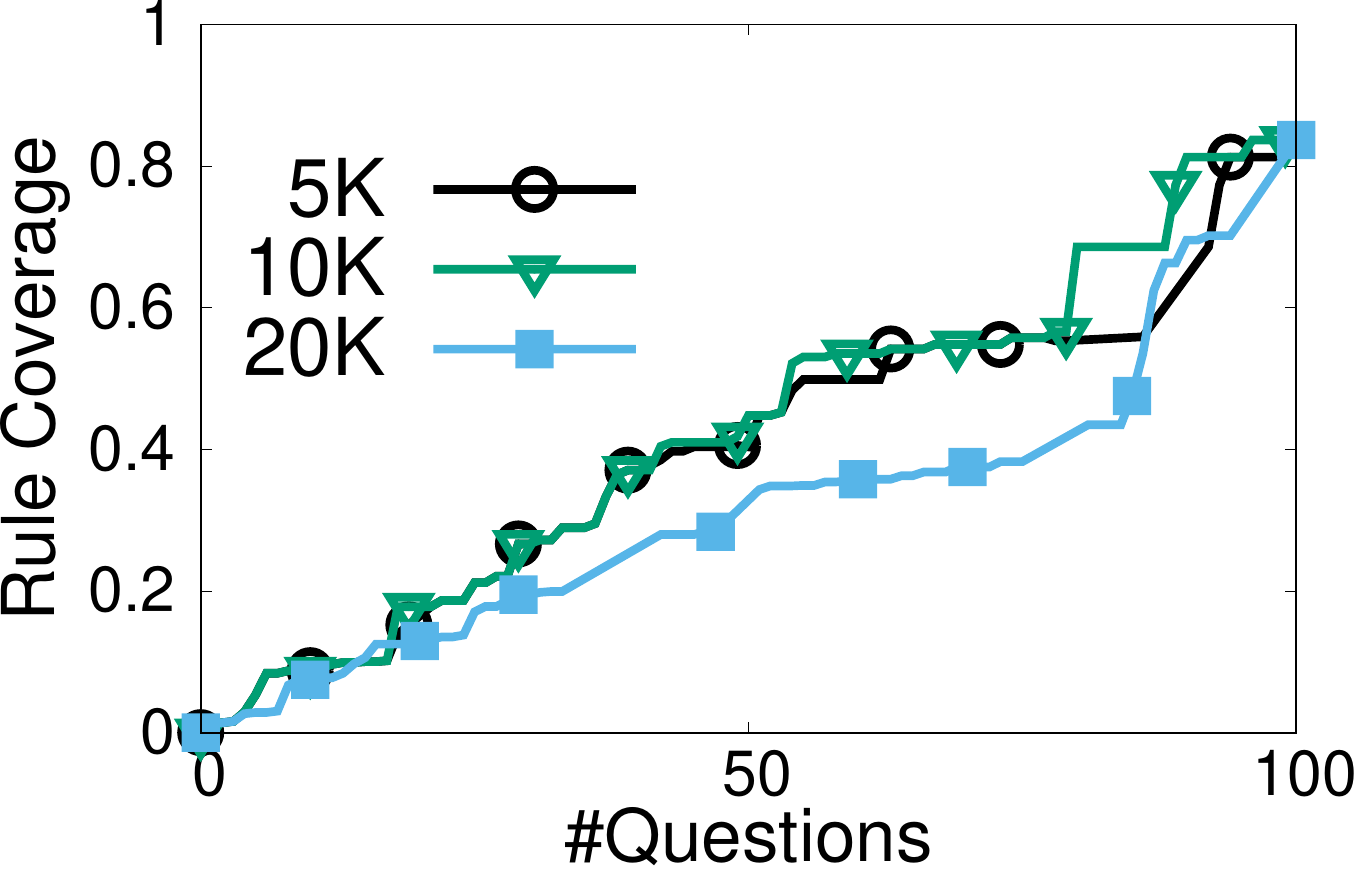}} ~
\caption{Sensitivity of \darwinp{HS}  to the number of candidates generated.\label{fig:indexsize}}
\end{figure}

\smallskip
\noindent\textbf{Sensitivity to number of generated candidates.} One of the parameters
of the \darwin\ framework is
the number of candidates that gets generated by the candidate-rule generation
component. In our experiments, \darwin{} generates 10K candidates rules with high
coverage and organizes them into an index. The goal is to make sure the set of
generated candidates contain some (if not all) of the precise rules. Choosing a
large value for the index size would satisfy this objective but affects the efficiency
and increases the number of candidates that \universal and \hybrid\ algorithms
need to consider. We observed that generating 10K candidates per iteration helped
\darwin{} identify precise candidate rules. Figure \ref{fig:indexsize} shows that
the performance of \darwinp{HS} algorithm is consistently similar for different
number of candidate rules generated.

\smallskip
\noindent\textbf{Sensitivity to classifier quality.} Figure \ref{fig:epoch} compares the
performance of \hybrid{} strategy on \texttt{musicians} dataset by varying the number of
epochs for which the neural network classifier is trained. With more epochs, the classifier tends to overfit more to the training data. We measure the number of questions from
\darwinp{HS} to the oracle in order to label at least 75\% of the positive sentences. It is
evident that the \darwin{} performance is robust to change in behavior of the classifier.

\begin{figure}
\centering
\includegraphics[width=0.24\textwidth]{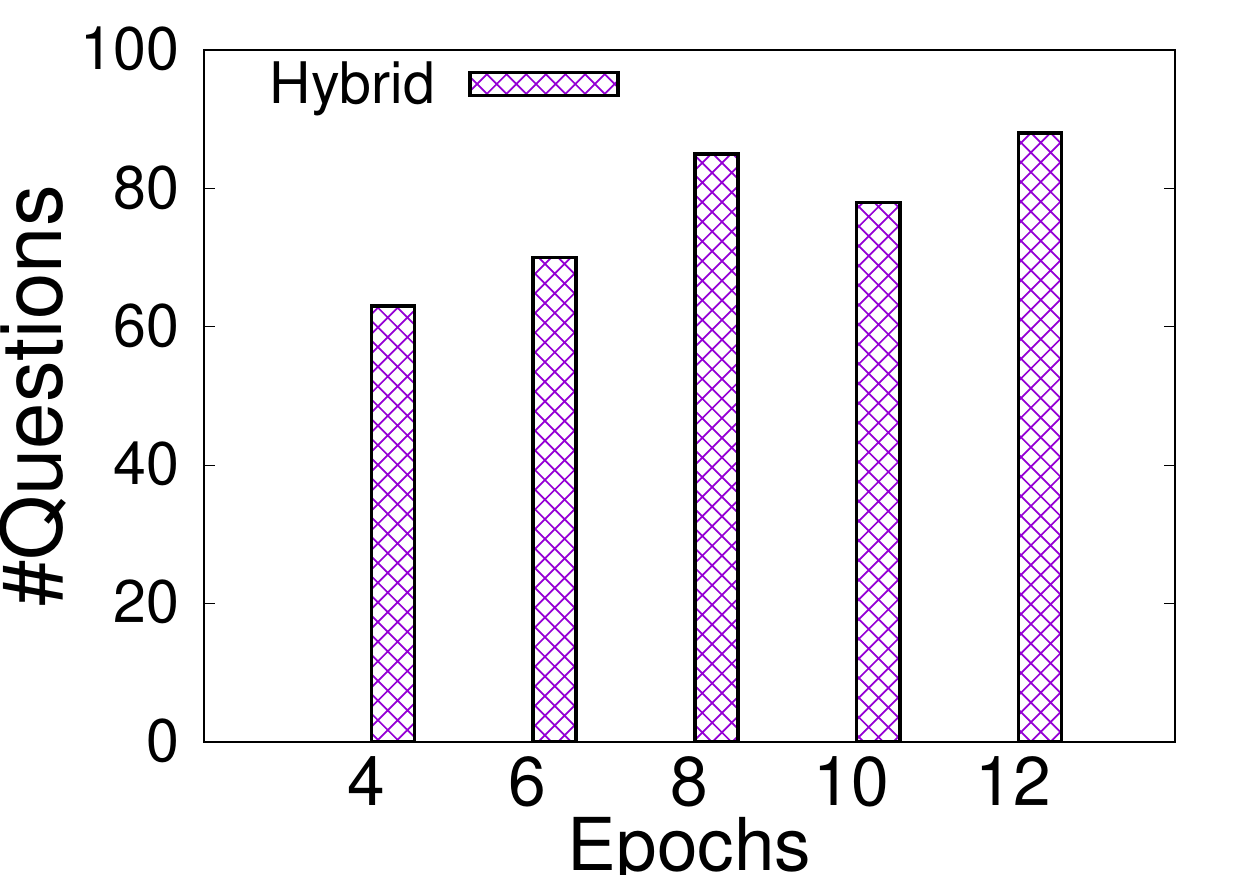}
\caption{Effect of classifier quality on \darwinp{HS} (on \texttt{musicians} dataset)\label{fig:epoch}}
\end{figure}

\end{document}